\documentclass[11pt]{article}
\usepackage{amsgen,amsmath,amstext,amsbsy,amsopn,amssymb,amsthm}
\usepackage{relsize}
\usepackage{graphicx}
\usepackage{subcaption}
\usepackage[hidelinks]{hyperref}
\usepackage{natbib}
\usepackage{apalike}
\usepackage{bibentry}
\usepackage{xcolor}

\usepackage{algorithm}
\usepackage{algpseudocode}
\usepackage{makecell}

\numberwithin{equation}{section}

\nobibliography*
\RequirePackage{bm}

\renewcommand{\cite}{\citep}

\newtheorem{theorem}{\bf Theorem}
\newtheorem{assumption}{Assumption}
\newtheorem{conditions}{\bf Conditions}

\newcommand{\ib}{\mathbf{i}}
\newcommand{\jb}{\mathbf{j}}

\newcommand{\lb}{\mathbf{l}}

\newcommand{\ub}{\mathbf{u}}
\newcommand{\vb}{\mathbf{v}}

\newcommand{\xb}{\mathbf{x}}

\newcommand{\Ab}{\mathbf{A}}

\newcommand{\Db}{\mathbf{D}}
\newcommand{\Eb}{\mathbf{E}}

\newcommand{\Ib}{\mathbf{I}}

\newcommand{\Lb}{\mathbf{L}}

\newcommand{\Ob}{\mathbf{O}}

\newcommand{\Qb}{\mathbf{Q}}

\newcommand{\Ub}{\mathbf{U}}
\newcommand{\Vb}{\mathbf{V}}
\newcommand{\Wb}{\mathbf{W}}
\newcommand{\Xb}{\mathbf{X}}
\newcommand{\Yb}{\mathbf{Y}}

\newcommand{\bbR}{\mathbb{R}}

\newcommand{\Sigmab}{\bm{\Sigma}}

\newcommand{\0}{{\mathbf{0}}}
\newcommand{\1}{{\mathbf{1}}}

\newcommand{\TS}{\textbf{TS}}

\newcommand{\astar}{^{\star}}
\newcommand{\astart}{^{\star\top}}

\newcommand{\overbar}[1]{\mkern 1.5mu\overline{\mkern-1.5mu#1\mkern-1.5mu}\mkern 1.5mu}

\title{Data Integration Via Analysis of Subspaces (DIVAS)}
\date{\today}
\author{Jack Prothero, Meilei Jiang, Jan Hannig,\\ Quoc Tran-Dinh, Andrew Ackerman, J.S. Marron}

\begin{document}

\maketitle
\newpage

\begin{abstract}
Modern data collection in many data paradigms, including bioinformatics, often incorporates multiple traits derived from different data types (i.e. platforms). We call this data multi-block, multi-view, or multi-omics data. The emergent field of data integration develops and applies new methods for studying multi-block data and identifying how different data types relate and differ. One major frontier in contemporary data integration research is methodology that can identify partially-shared structure between sub-collections of data types. This work presents a new approach: Data Integration Via Analysis of Subspaces (DIVAS). DIVAS combines new insights in angular subspace perturbation theory with recent developments in matrix signal processing and convex-concave optimization into one algorithm for exploring partially-shared structure. Based on principal angles between subspaces, DIVAS provides built-in inference on the results of the analysis, and is effective even in high-dimension-low-sample-size (HDLSS) situations.
\end{abstract}
\newpage

\section{Introduction}\label{sec:intro}

Modern experiments are increasingly likely to produce complex data derived from multiple sources. One common example is a single group of $n$ \textit{objects} (i.e. cases, observations, patients) being observed across $K$ different \textit{views} or \textit{data blocks}, each with their own sets of $d_k,\ k=1,\ldots,K,$ \textit{traits} (i.e. variables, features, descriptors) and measurement methodologies. We call data collected and organized this way \textit{multi-block data}. Simple data analysis approaches would concatenate multiple data blocks into a single data matrix of $n$ data objects and $d_1+\cdots+d_K$ traits. However, these approaches ignore the often important relationships between the views and obscure insights on which information comes from which data block. We specifically aim to address this challenge by considering the data blocks as separate units and searching for \textit{shared structure} between them. Shared structure can be defined in several ways. Here and in some other contemporary approaches seen in Section~\ref{sec:litreview} it means common modes of variation modeled by low-rank matrices. More mathematical details of this modeling choice can be found in Section~\ref{sec:methodology}. Our proposed method, \textit{Data Integration Via Analysis of Subspaces} (DIVAS), incorporates state-of-the-art advances in matrix perturbation theory and optimization to provide insights about both shared and partially-shared joint structure between several data blocks.

DIVAS gives a novel approach for finding structure in a multi-block data set based on searching for shared subspaces between different collections of data blocks. The data blocks are intrinsically linked by having the same \textit{trait space} $\bbR^n$, so we primarily search for shared subspaces within the trait space. In contrast to other approaches, angles form the foundation of our analysis of the relationships between these subspaces. In particular, \textit{principal angles} are the measure of choice of proximity between subspaces. These angles are applied in a rigorous framework of inference that provides relevant statistical significance determinations for the chosen subspaces.

Subspaces of trait space have corresponding induced subspaces of \textit{object space} $\bbR^{d_k}$ for each data block. Combined together, these subspaces can be decomposed into \textit{modes of variation}. A mode of variation is a rank 1 matrix formed from the outer product of two vectors: one in object space and one in trait space. In the terminology of Principal Components Analysis (PCA) these would be a loadings (direction) vector and a scores vector, respectively. Considering subspaces in terms of modes of variation is particularly useful for visualization. The scores vectors demonstrate relationships between the data objects, and the loadings vectors provide information about which traits are driving the variation. The ultimate result of a DIVAS exploration of a data set is a set of modes of variation for each data block associated with each block collection.

The rest of the paper proceeds as follows. The remainder of Section~\ref{sec:intro} continues with background information on data integration in general. Section~\ref{sec:methodology} details DIVAS methodology and demonstrates the method's performance on a synthetic data set.Section~\ref{sec:genomics} provides a prototypical application of DIVAS in cancer genomics.Section~\ref{sec:data} contains a case study on twentieth century mortality using DIVAS. Section~\ref{sec:conclusion} summarizes some brief conclusions. The Appendices~\ref{sec:mp}, \ref{sec:paa}, \ref{sec:noise}, and \ref{sec:optdetails} provide reference materials on random matrix theory, principal angle analysis, residual matrix estimation, and details on the optimization problem solved by DIVAS, respectively. MATLAB code is available for download at \url{https://github.com/jbprothero/DIVAS2021}.

\subsection{Data Integration Literature}\label{sec:litreview}

A time-honored multi-block data analysis method is Canonical Correlations Analysis (CCA), proposed by \citep{cca}. Given two blocks of data $\Xb$ and $\Yb$, CCA seeks to maximize the Pearson correlation between vectors from the span of each data block in trait space. 
The fundamental ideas of CCA have been thoroughly extended to more general settings. Several authors propose different generalizations of CCA for locating highly correlated structure between three or more data blocks, including \citep{horst1961mcca, gcca,nielsen2002mcca}. Some others have experimented with kernel methods for CCA, as in \citep{kcca} and \citep{cai2017kcca}.

In the context of machine learning, methods like CCA are termed \textit{multi-view methods}. Multi-view learning broadly includes methods like co-training for semi-supervised learning \citep{cotraining}, SVM-2K \citep{svm2k}, subspace learning \citep{mvsubspace}, and other multi-view extensions of paradigms like active learning and ensemble learning. See \citet{mvsumshort}, \citet{mvsumlong}, and \citet{li2019survey} for more details on these extensions.



Any CCA-based method is ultimately focused on finding jointly shared structure between each available block or view of the data. Oftentimes, given low-rank approximations of each data block we are also interested in a full factorization of the signal present in each block into a joint component shared between all blocks and an individual component unique to each block. One algorithm that produces such a factorization is Joint and Individual Variation Explained (JIVE) \citep{jive}. After initially choosing a signal rank for each data matrix using a permutation testing approach, the algorithm seeks to minimize residual energy by alternating between determining joint structure and individual structure. Broadly the algorithm accomplishes its goal, but the optimization problem can proceed slowly and there is no underlying inferential justification for the chosen boundary between joint and individual structure.


A later generation of JIVE, dubbed Angle-based JIVE (AJIVE) \citep{ajive}, was proposed to address the above shortcomings. Selection of joint structure happens in a quick, single step based on principal angle analysis and the delineation between joint and individual structure is based on a bound on the angles between original and perturbed subspaces found in \citet{wedin}. Initial rank selection, however, is performed ad hoc in a separate initial step, and as described in \citet{ajive} the perturbation angle bounds used can become extremely conservative under rank mis-specification. Additionally, neither JIVE nor AJIVE consider partially-shared joint structure between subsets of blocks.





Decomposition of data blocks into partially-shared joint structure components is one of the primary frontiers in contemporary data integration research. Two approaches to the problem, \citep{gaynapartial} and \citep{bass}, both model partially-shared information via structured sparsity in a basis matrix for the concatenated data blocks. \citet{gaynapartial} determine sparse structure via bi-cross-validation, while \citet{bass} determine sparse structure via a collection of Bayesian priors. A third approach is found in the recent unpublished manuscript \citep{yi2022pss}. Their method models partially-shared information via subspace intersections and a hierarchical matrix nuclear norm regularization scheme. This formulation eschews factorization into scores and loadings subspaces entirely in order to work with a convex optimization problem, to capture potentially non-orthogonal partially-shared structure across block collections, and to ensure identifiability. DIVAS maintains identifiability even among potentially non-orthogonal partially-shared structures via a sequential search through each collection of data blocks. DIVAS also involves a more challenging non-convex optimization problem based on factorized matrices, but in doing so achieves relevant statistical significance measurements for the resulting shared and partially-shared subspaces.





Another recent pursuit in data integration research is complete incorporation of all the information from the data blocks. In most cases the data blocks share the same number of data objects, so integrative analysis often takes place primarily in trait space, with corresponding information about the contributions of certain traits to shared structure being determined subsequently. In many cases, information about the contributing traits is just as pertinent as the shared structure itself, including if data blocks are bi-dimensionally linked as in \citep{lock2020bidimensional} or bi-dimensionally matched as in \citep{yuan2021doublematched}. This is often the case in bioinformatics where the traits represent measurements of particular genes and the primary goal is to identify genes or other biological factors that contribute to patterns observed across the data blocks. The above papers each propose their own method for incorporating trait information in the analysis in the situations where the data blocks are appropriately linked. Another methodology found in \citep{shu2021cdpa} attempts to incorporate trait information for more general multi-block situations, but it relies on a computationally-taxing row-matching algorithm as part of its procedure and the method cannot parse partially-shared joint structure. Our method utilizes subspace perturbation theory that applies along either dimension of the data blocks and in any scenario with mutli-block data. This makes trait information easy to incorporate throughout the algorithm that locates partially-shared structure.


\section{Methodology}\label{sec:methodology}

Let $\Xb_1,\dots,\Xb_K$ be data blocks each containing the same set of $n$ data objects and distinct sets of $d_1,\dots, d_K$ traits. In matrix calculations we follow the bioinformatics convention where matrix columns are data objects and matrix rows are traits (i.e. the matrix $\Xb_k$ has $d_k$ rows and $n$ columns). In our data model, shown in \eqref{eq:datamodel}, each data block is assumed to be a low-rank signal matrix $\Ab_k$ plus a full-rank noise matrix $\Eb_k$: 
\begin{equation}
\label{eq:datamodel}
\Xb_k = \Ab_k + \Eb_k.
\end{equation}
We assume each entry of $\Eb_k$ is independent with identical variance $\sigma^2$ and finite fourth moment. For inferential purposes, we also assume \textit{rotational invariance} between signal and noise data matrices. The mathematical details of this assumption are provided in Section~\ref{sec:rotboot} where they are maximally relevant. 

Further discussion of the modeling assumptions for DIVAS requires notation for describing different collections of data blocks in detail. Consider the power set $2^{\{1,\dots,K\}}$ as a set of index sets, where each element $\ib\in2^{\{1,\dots,K\}}$ represents a particular collection of data block indices. Each $\ib$ indexes a shared structure among the data blocks $\Xb_k$ with $ k\in\ib$. Denote by $|\ib|$ the cardinality of $\ib$.


In order to define structure partially-shared across the data blocks, we decompose the signal matrices $\Ab_k,\ k=1\ldots,K$ into a sum of low-rank signal matrices, each of which corresponds to the joint structure shared between the collection of blocks indicated by the index $\ib$:
\begin{equation}
    \label{eq:signalmodel}
    \Ab_k = \sum_{\ib | k\in\ib} \Lb_{\ib,k}\Vb_{\ib}^\top.
\end{equation}
Here the sum extends over all index sets $\ib\in 2^{\{1,\ldots,K\}}$ that satisfy $k\in\ib$.
The $n\times r_\ib$ \textit{scores matrices} $\Vb_{\ib}$ 
model the shared structure in the trait space between the data blocks $\Xb_k$ with $k\in\ib$. The $d_k\times r_\ib$ \textit{loadings matrices} $\Lb_{\ib,k}$  contain the induced object space structure in each block $\Xb_k$ with $k\in\ib$.
In order to ensure identifiability of the decomposition   \eqref{eq:signalmodel}, the factorized signal matrices $\Lb_{\ib, k}$ and $\Vb_{\ib}$ are required to satisfy the following conditions: (Here and for the rest of the manuscript the notation $\left[\Vb_{\ib} \right]_{\ib|\ib\in S}$ denotes horizontal matrix concatenation $\left[ \Vb_{\ib_1} \cdots \Vb_{\ib_{|S|}} \right]$ of all matrices $\Vb_{\ib}$ with $\ib\in S$.)
\begin{conditions}\label{asump:decomposition}
Identifiability conditions for decomposition
\eqref{eq:signalmodel}:
\begin{enumerate}
    \item The columns of each $\Vb_{\ib}$ are orthonormal.
    \item For two different block index sets $\ib\neq\jb$, if $\ib\subset\jb$ or $\jb\subset\ib$, then the subspaces spanned by the columns of $\Vb_{\ib}$ and $\Vb_{\jb}$ in the trait space are orthogonal.
    \item The matrix $\left[\Vb_{\ib}\right]_{\ib|\ib\in 2^{\{1,\dots , K\}}}$, concatenated over all $\ib\in 2^ {\{1,\ldots,K\}},$ has rank equal to its number of columns.
    \item  For all $k$, the matrix $\left[\Lb_{\ib,k}\right]_{\ib\,|\,k\in\ib}$, concatenated over all $\ib\in 2^ {\{1,\ldots,K\}}$ so that $k\in\ib$, has rank equal to its number of columns.
\end{enumerate}
\end{conditions}
Note that the columns of the loadings matrices $\Lb_{\ib, k}$ are not required to be orthogonal and may have arbitrary magnitude in order to encode scale information. The $d_k\times n$ matrix $\Ab_{\ib,k}=\Lb_{\ib,k}\Vb_{\ib}^\top$ has rank $r_{\ib}$, the number of columns of $\Vb_\ib$. The matrix $\Ab_{\ib,k}$  will be called the {\em partially shared joint structure} between blocks in $\ib$. When $\ib$ is a singleton we also call it {\em individual structure}, and when $\ib=\{1,\ldots,K\}$ we also call it {\em fully joint structure}.

Next we prove existence and uniqueness of the joint structure decomposition \eqref{eq:signalmodel} in the absence of noise:
\begin{theorem}\label{th:construct}
For a set of signal matrices $\Ab_1,\dots ,\Ab_K$, there exists a set of matrices $\Lb_{\ib,k}, \Vb_{\ib}$ satisfying \eqref{eq:signalmodel} and identifiability Condition~\ref{asump:decomposition}. The joint structure matrices $\Ab_{\ib,k}=\Lb_{\ib,k}\Vb_{\ib}^\top$ are uniquely determined for all $\ib\in 2^{\{1,\ldots,K\}}$ and $k\in\ib$.
\end{theorem}
\begin{proof}
We proceed with a constructive proof by induction. The steps found here will be also used in the development of the estimation algorithm in Section~\ref{sec:jse}. Denote by $\mathcal{V}_{\ib}$ the intersection of the subspaces spanned by the columns of transposed signal matrices $\Ab_k^\top$ with $k\in\ib$ in trait space. 

Step 1: Consider the index set $\ib=\{1,\ldots,K\}$. Choose $\Vb_{\{1,\dots , K\}}$ such that its columns form an orthonormal basis for $\mathcal{V}_{\{1,\dots , K\}}$.  Clearly the relevant parts of Condition~\ref{asump:decomposition} are satisfied. Also notice that any vector orthogonal to $\Vb_{\{1,\dots , K\}}$ is in at most $K-1$ subspaces spanned by the columns of  $\Ab_k^\top$.

Step 2: Let us assume that we have defined $\Vb_{\ib}$ for all $K\geq |\ib|\geq q$ that satisfy Condition~\ref{asump:decomposition}, and,
for any $|\ib|=q$, any vector orthogonal to $\left[\Vb_{\jb}\right]_{\jb\supset\ib}$ is included in at most $q-1$ subspaces spanned by the columns of $\Ab_k^\top$, $k\in\ib$.
For any $\ib$ with $|\ib|=q-1$ select  $\Vb_{\ib}$ such that its columns form an orthonormal basis for $\mathcal{V}_{\ib} \cap \left[\Vb_{\jb}\right]_{\jb\supset\ib}^\perp$, the part of the space $\mathcal{V}_{\ib}$ orthogonal to all $\Vb_{\jb}$, $\jb\supset\ib$.
Each $\Vb_{\ib}$ chosen this way satisfies parts 1 and 2 of Condition~\ref{asump:decomposition} by construction.
Condition~3 is satisfied, because if there was a rank deficiency in a concatenated matrix $\left[\Vb_{\ib}\right]_{\ib\,|\, |\ib|\geq q-1}$ there would be two indices $\ib$, $\jb$ such that $\ib\neq\jb$ and the spans of the matrices $\Vb_{\ib}$ and $\Vb_{\jb}$ share some vector in common. However this vector would already be included in $\Vb_{\ib\cup\jb}$ which is a contradiction.
Finally, for any $|\ib|=q-1$, any vector orthogonal to $\left[\Vb_{\jb}\right]_{\jb\supset\ib}$ is included in at most $q-2$ subspaces spanned by the columns of $\Ab_k^\top$, $k\in\ib$.
This completes the inductive construction of the collection of $\Vb_{\ib}$.

For each $k$, notice that the column span of $\left[\Vb_{\ib}\right]_{\ib|k\in\ib}$ is $\mathcal V_{\{k\}}$, the space spanned by columns of $\Ab_k^\top$. Thanks to part~3 of Condition~\ref{asump:decomposition}, the matrices $\Lb_{\ib,k}$ are chosen as the unique solution of the equation formed using the concatenated matrices
$\left[\Lb_{\ib,k}\right]_{\ib|k\in\ib}\cdot \left[\Vb_{\ib}\right]_{\ib|k\in\ib}^\top=\Ab_k.$ These $\Lb_{\ib,k}$ satisfy part~4 of Condition~\ref{asump:decomposition} by construction.

Now assume there exists some other collection
$
    \Ab_k = \sum_{\ib | k\in\ib} \tilde{\Lb}_{\ib,k}\tilde{\Vb}_{\ib}^\top
$
also satisfying Condition~\ref{asump:decomposition}. Following  similar arguments as above we see that the column spaces of $\tilde{\Vb}_{\ib}$ and ${\Vb}_{\ib}$ are the same, and therefore there exists an orthonormal $r_{\ib}\times r_{\ib}$  matrix $\Qb_{\ib}$, so that  $\tilde{\Vb}_{\ib}=\Vb_{\ib}\Qb_{\ib}$. Consequently
$\tilde{\Lb}_{\ib, k}=\Lb_{\ib, k}\Qb_{\ib}$ and
$\tilde{\Lb}_{\ib,k}\tilde{\Vb}_{\ib}^\top=\Lb_{\ib,k}\Vb_{\ib}^\top$.
\end{proof}

The fact that the matrices $\Lb_{\ib, k}$ and $\Vb_{\ib}$ are only determined up to basis rotation is a natural result of DIVAS being a subspace based method and not focused on matrices. In particular,
the most important information contained in $\Lb_{\ib, k}$ and $\Vb_{\ib}$ is the subspaces their columns span in object space and trait space respectively. This allows DIVAS to efficiently handle near equal singular values that could cause problems for matrix based approaches. For interpretive purposes, it can be helpful to choose a particular informative basis for the shared subspaces and examine modes of variation of the data along those basis directions. These modes of variation of the data may be formed by outer-multiplying corresponding columns of suitably rotated $\Lb_{\ib,k}$ and $\Vb_{\ib}$. In Section~\ref{sec:recon} we discuss a particular choice of such informative basis rotations $\Qb_\ib$ obtained using an SVD of a particular estimated signal matrix.

Throughout the description of methodology we will use the following synthetic three-block data example to illustrate each step of DIVAS. Each block includes a different set of traits associated with 400 observations. To mimic challenging data situations with large disparities in trait set sizes, Block 1 has 200 traits, Block 2 has 400 traits, and Block 3 has $10000$ traits. Figure~\ref{3blocktruth} displays this synthetic data set using matrix \textit{heatmaps}. Heatmaps are a graphical display of matrix entry magnitude using color. Negative entries are shown in shades of blue and positive entries are shown in red, with color saturation indicating the magnitude of each entry. This means that entries close to zero are shown with a low-saturation white color. The color scaling ranges of each heatmap are shown in the color bar below each individual plot.

Each row demonstrates the formation of one of the data blocks via the data model in \eqref{eq:datamodel}. The left-most column of heatmaps shows the observed data blocks $\Xb_k$, and the right-most column of heatmaps shows the noise matrices $\Eb_k$, which in this example are i.i.d. Gaussian matrices. The middle columns display the various rank 1 components that sum to each block's signal matrix $\Ab_k$ as in \eqref{eq:signalmodel}. Each signal matrix is comprised of a fully-shared component (second column) and partially-shared components between two of the three data blocks (third, fourth, and fifth columns). As required by our model, these fully-shared and partially-shared components are constructed such that the trait space subspace of the fully-shared component is orthogonal to the corresponding trait space subspaces of each partially-shared component. However, the partially-shared subspaces are not mutually orthogonal. In fact, each pair of partially-shared trait subspaces each has a principal angle of 60 degrees between them in the trait space, $\bbR^{400}$. Adding the matrices in the middle columns across each row in the manner of \eqref{eq:signalmodel} combines the signal components into rank 3 signal matrices. Adding the noise matrices in the manner of \eqref{eq:datamodel} then produces the observed data blocks (first column).


\begin{figure}[h]
\centering
\includegraphics[scale=0.51]{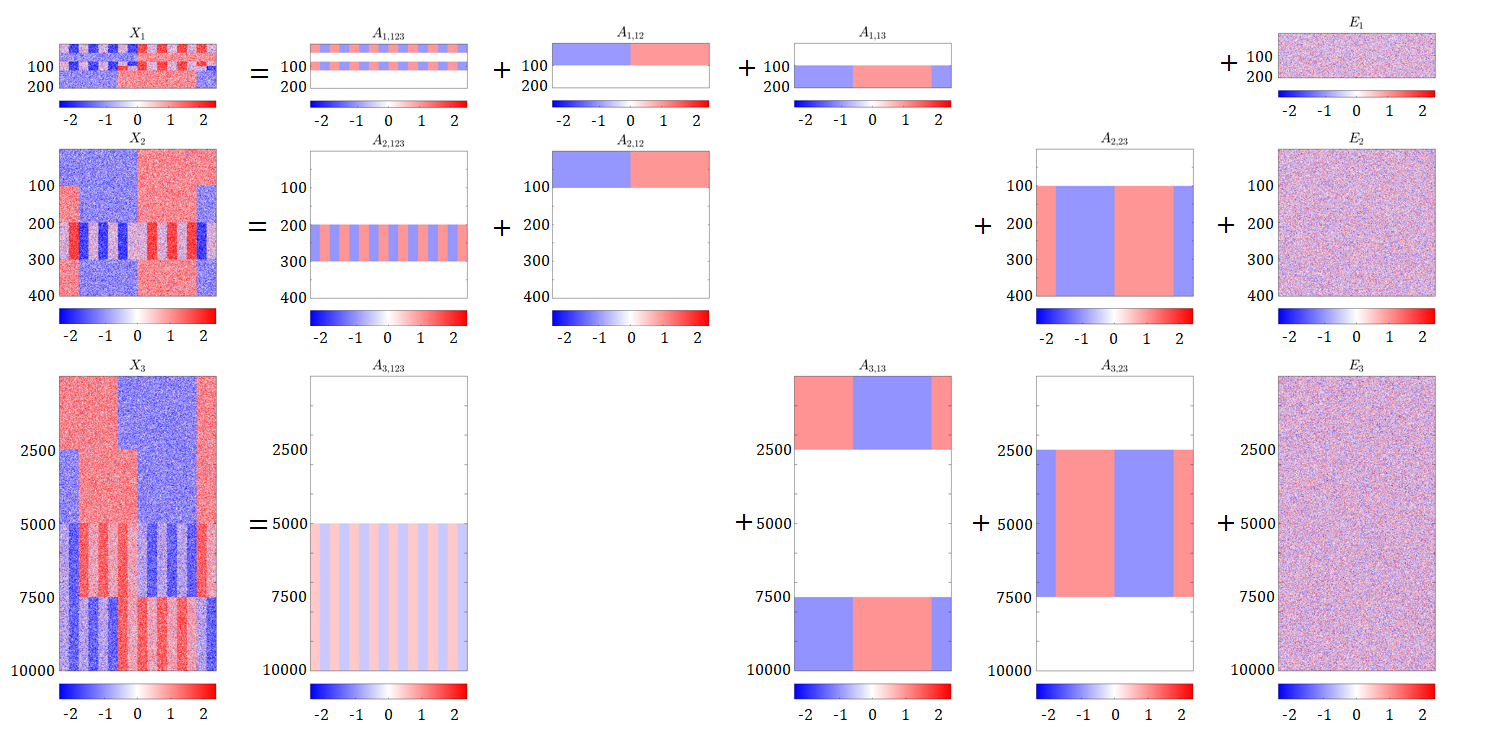}
\caption{Heatmap view of $K=3$ synthetic example construction. Heatmaps share a common color scheme displayed in the color bars below each plot, with white representing 0 magnitude. The three blocks in the first column are the observed data, and are formed by adding up the other matrices in their respective row. The three blocks in the second column show the rank 1 fully-shared structure common to each block. The next three columns show the rank 1 partially-shared structure common to each subset of two blocks. The final column shows the noise matrices for each data block.}\label{3blocktruth}
\end{figure}

The procedure of DIVAS takes the observed data blocks as input and outputs a full breakdown of shared and partially-shared structure between them over three steps. The first step, described in Section~\ref{sec:sse}, extracts and estimates the dimension, magnitude, and direction of each block's signal subspace. The second step, described in Section~\ref{sec:jse}, combines the information from each block to locate shared directions between subspaces. The third step, described in Section~\ref{sec:recon}, uses those shared directions to form estimates $\hat{\Ab}_{\ib,k}$ of the partially shared joint structure matrices for each data block, e.g., estimates of the middle columns of Figure~\ref{3blocktruth}. Section~\ref{sec:graphics} describes the visual display of DIVAS output and corresponding diagnostic measurements.

\subsection{Signal Subspace Extraction}\label{sec:sse}


The first step of DIVAS is to estimate the subspaces spanned in both object space and trait space by the signal matrix for each data block. As the properties and analyses that take place in this subsection apply to each data block independently, in this section we suppress the subscript $k$ indexing data blocks $\Xb_k$ when referring to data-block-specific quantities for simplicity of notation. The indexing subscript will return when information from different data blocks is combined together in Section~\ref{sec:jse}. When we observe $\Xb$, we are observing data that's been perturbed in both magnitude (singular values) and orientation (basis vectors) from the signal $\Ab$. The signal magnitude is readily recoverable from the data magnitude via the signal extraction procedure described subsequently in Section~\ref{sec:mpsignal}. The signal orientation is itself more challenging to estimate from the data orientation; there's no reason to favor one direction of rotation over another under the rotational invariance assumptions on the noise matrix $\Eb$. However, the key to DIVAS is to quantify a range of feasible signal orientations given the observed data orientation using bounds on principal angles. These techniques are described in Sections~\ref{sec:perturb} and \ref{sec:rotboot}. The results of this step for the synthetic data set presented in Figure~\ref{3blocktruth} are shown in Section~\ref{sec:syntheticsse}

\subsubsection{Signal Subspaces}\label{sec:mpsignal}


In \citet{sn2013}, Proposition 5 demonstrates that if the noise in \eqref{eq:datamodel} is orthogonally invariant, any procedure for extracting signal from a data matrix $\Xb$ need only consider the singular values of the data matrix.  In Section~8.2 of \citet{gd2017}, the authors provide a similar result under the conditions that the noise is i.i.d.~with zero mean and finite fourth moment and the left and right singular vector matrices (i.e. the subspace orientation information) of the signal matrix are drawn uniformly at random.
Motivated by this, our algorithm uses the signal extraction procedures developed in these articles. We perform SVD on $\Xb$ to find $\Xb = \overbar{\Ub}\overbar{\Db}\overbar{\Vb}^{\top}$. The columns of the matrices $\overbar{\Ub}$ and $\overbar{\Vb}$ are orthonormal bases for the subspaces spanned in object space and trait space of $\Xb$ respectively, and the diagonal entries of $\overbar{\Db}$ are the singular values of $\Xb$. Denote these singular values as $\overbar{\nu}_1,\dots,\overbar{\nu}_{d\wedge n}$. Estimations of the signal matrix $\hat{\Ab}$ typically take the form of a decomposition in terms of rank 1 matrices/approximations that combine to form the estimated object space and trait space subspaces. The vectors $\overbar{\ub}_i$ and $\overbar{\vb}_i$ denote the $i$th columns of the matrices $\overbar{\Ub}$ and $\overbar{\Vb}$ respectively, and $\eta(\bullet)$ is a function from $\bbR^+$ to $\bbR^+$ for shrinking the singular values:
\begin{equation}
\label{eq:signalest}
\hat{\Ab} = \sum_{i=1}^{d\wedge n} \eta(\overbar{\nu}_i) \overbar{\ub}_i \overbar{\vb}_i^{\top}.
\end{equation}


Common choices for $\eta$ include \textit{soft thresholding}: $\eta_{soft}(\nu) = (\nu-c)\vee 0$, and \textit{hard thresholding}: $\eta_{hard}(\nu) = \nu \mathbb{I}_{\{\nu\geq c\}}$, for some constant $c$, and where $\mathbb{I}_{\{\bullet\}}$ represents an indicator function. In either case, any singular value smaller than $c$ is set to 0. This means both procedures have dimension-reducing effects on the estimated $\hat{\Ab}$, and only subspaces $\overbar{\ub}_i \overbar{\vb}_i^{\top}$ associated with nonzero transformed singular values contribute to the estimate. \citet{gd2014} outline optimal choices for $c$ for both soft and hard thresholding in terms of the \textit{aspect ratio} $\beta = \frac{d\wedge n}{d\vee n}$ and the standard deviation $\sigma$ of the noise matrix $\Eb$. 

In the additive noise data matrix model \eqref{eq:datamodel}, the presence of noise inflates the singular values associated with the signal component of the data matrix. Hard thresholding does not account for this phenomenon at all and soft thresholding often overcorrects by applying the same amount of shrinkage to each nonzero singular value. \citet{sn2013} and \citet{gd2017} each propose optimal thresholding functions based on the Marchenko-Pastur distribution (see Appendix~\ref{sec:mp}) under a variety of matrix norms. We use the operator-norm-optimal function $\eta^*$ from \citet{gd2017} for DIVAS:
\begin{equation}
\label{eq:opShrink}
\eta^*(\nu) = \begin{cases}
\frac{1}{\sqrt{2}}\sqrt{\nu^2 - \beta - 1 + \sqrt{(\nu^2 - \beta - 1)^2 - 4\beta}}, & \nu \geq 1 + \sqrt{\beta}; \\
0, & \nu < 1 + \sqrt{\beta}.
\end{cases}
\end{equation}
Figure~\ref{threshes} demonstrates how this shrinkage function \eqref{eq:opShrink}, blue solid line, compromises between soft and hard thresholding for singular values with different magnitudes for a matrix with aspect ratio $\beta = 1$ and noise standard deviation $\sigma = 1$. Small values are thresholded according to optimal soft thresholding (magenta dot-dash line), but the shrinkage function approaches optimal hard thresholding (black dashed line) for larger values.

\begin{figure}[h]
\centering
\includegraphics[scale=0.4]{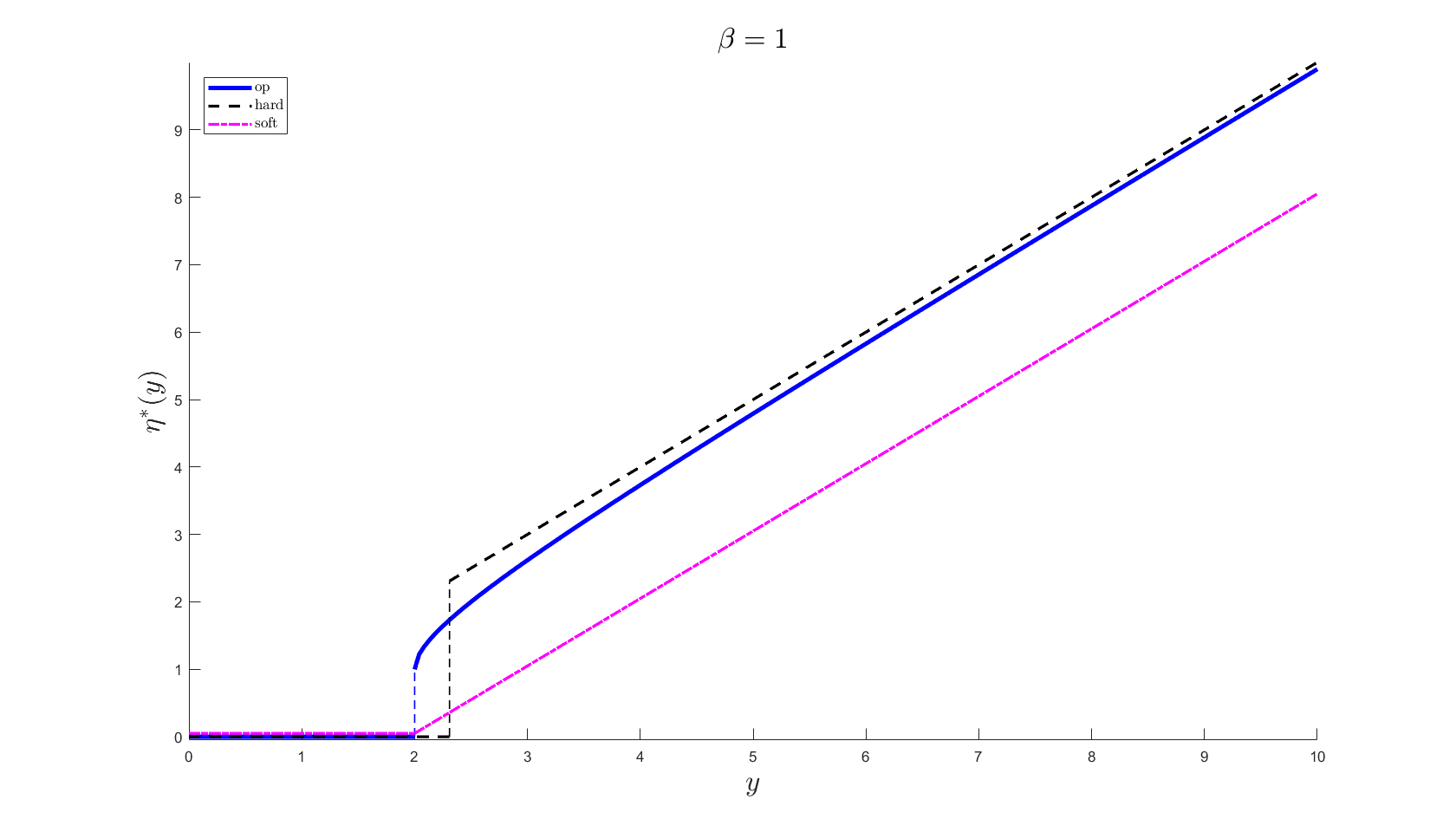}
\caption{Functions for hard thresholding, soft thresholding, and optimal shrinkage under operator norm loss for a square matrix ($\beta=1$). The optimal shrinkage function compromises between the other two approaches. Figure produced with code from \cite{gd2017}.}\label{threshes}
\end{figure}

Equation~\eqref{eq:opShrink} assumes noise standard deviation $\sigma = 1$. To use the shrinkage function in general settings, we must appropriately scale the singular values before and after shrinkage according to some estimate of the standard deviation of the noise $\hat{\sigma}$. \citet{sn2013} use a grid search over several candidate values for $\hat{\sigma}$ to find a value that minimizes the Kolmogorov-Smirnov distance between the non-signal singular values and the appropriate Marchenko-Pastur distribution. \citet{gd2017} opt for the simple, robust, closed-form estimate $\hat{\sigma} = \frac{\nu_{median}}{\sqrt{MP(\beta)_{0.5}}}$, where $\nu_{median}$ denotes the median singular value of $\Xb$ and $MP(\beta)_{0.5}$ denotes the median of the Marchenko-Pastur distribution with parameter $\beta$ (see Section~\ref{sec:mp}). We use the latter method for DIVAS noise standard deviation estimation.

Combining the previous equations, our estimate for the signal matrix $\hat{\Ab}$ for a given data block $\Xb$ is:
\begin{equation}
\label{eq:finala}
\hat{\Ab} =  \sum_{i=1}^{d\wedge n} \hat{\sigma}\eta^*(\overbar{\nu}_i/\hat{\sigma}) \overbar{\ub}_i \overbar{\vb}_i^{\top}.
\end{equation}
Let $\hat{\nu_i} = \hat{\sigma}\eta^*(\overbar{\nu}_i/\hat{\sigma})$ be the $i$th shrunken singular value of $\Xb$. Let $\hat{r}$ be the number of nonzero shrunken singular values, and therefore the estimated rank of $\Ab$. Let $\hat{\Ub}$ and $\hat{\Vb}$ be matrices containing the first $\hat{r}$ columns of $\overbar{\Ub}$ and $\overbar{\Vb}$, respectively. Using this notation and defining the matrix $\hat{\Db}$ as the $\hat{r}\times\hat{r}$ diagonal matrix with diagonal entries equal to $\hat{\nu}_1,\dots,\hat{\nu}_{\hat{r}}$, we can also write $\hat{\Ab} = \hat{\Ub}\hat{\Db}\hat{\Vb}^{\top}$. Note that $\hat{\Ub}$ is therefore an orthonormal basis for the subspace spanned in object space of $\hat{\Ab}$ and $\hat{\Vb}$ is an orthonormal basis for the subspace spanned in trait space of $\hat{\Ab}$.

\subsubsection{Angle Perturbation Theory} \label{sec:perturb}


The foundation of DIVAS is determining whether candidate directions $\vb\astar\in \bbR^n$ lie in the trait space span $\TS(\Ab)$ of the signal matrix $\Ab$. If $\Ab$ was observable, this would simply amount to checking whether the angle $\theta$ between $\vb\astar$ and $\TS(\Ab)$ was 0. Since $\Ab$ and $\theta$ are unobservable, we aim to estimate $\theta$ based on the observable estimated low-rank signal matrix $\hat{\Ab}$ and an estimate of the noise variation $\hat{\Eb}$ developed in Appendix~\ref{sec:noise}. Specifically, we want to choose a \textit{perturbation angle bound} $\hat{\phi}$ that defines a cone-shaped significance region around $\TS(\hat{\Ab})$ which contains $\TS(\Ab)$ with high probability. Directions $\vb\astar$ lying within that significance region would then be potential basis directions for $\TS(\Ab)$. Hence directions $\vb\astar$ that lie within the significance regions of multiple data blocks would then be potential basis directions for partially-shared joint structure between those blocks. To arrive at such a data-block-wise uniform perturbation angle bound $\hat{\phi}$, we first look to bound the range of possible values for $\theta$ for one given candidate direction $\vb\astar$.





The first step is construction of the range of values for $\theta$ based on the relationships between the projections of $\vb\astar$ onto various subspaces. This is illustrated using a simple low-dimensional example in Figure~\ref{thetastars}. In particular, $\vb\astar$ (green vector in both panels of Figure~\ref{thetastars}) is projected onto each of $\TS(\Ab)$ (translucent purple plane) and $\TS(\hat{\Ab})$ (solid gold plane), with those projections denoted $\vb\astar_{proj}$ (red solid lines) and $\hat{\vb}\astar_{proj}$ (blue solid lines) respectively. Computations of these projections are based on the orthonormal basis matrices $\Vb$ for $\TS(\Ab)$ and $\hat{\Vb}$ for $\TS(\hat{\Ab})$. Let $\hat{\theta}$ be the angle between $\vb\astar$ and $\TS(\hat{\Ab})$. We can write expressions for $\theta$ and $\hat{\theta}$ in terms of the above quantities as follows: $$\theta = \arccos\left(\frac{\langle\vb\astar,\vb\astar_{proj}\rangle}{\parallel \vb\astar\parallel\parallel\vb\astar_{proj}\parallel}\right);\;\; \vb\astar_{proj} = \Vb\Vb^{\top}\vb\astar.$$
$$\hat{\theta} = \arccos\left(\frac{\langle\vb\astar,\hat{\vb}\astar_{proj}\rangle}{\parallel \vb\astar\parallel\parallel\hat{\vb}\astar_{proj}\parallel}\right);\;\; \hat{\vb}\astar_{proj} = \hat{\Vb}\hat{\Vb}^{\top}\vb\astar.$$

We can construct bounds involving $\theta$ and $\hat{\theta}$ by considering further projections between $\TS(\Ab)$ and $\TS(\hat{\Ab})$. The total angle traversed by projecting $\vb\astar$ to $\TS(\Ab)$ and then projecting that result, $\vb\astar_{proj}$, onto $\TS(\hat{\Ab})$ (red dashed line in left panel of Figure~\ref{thetastars}) is at least as large as $\hat{\theta}$, the angle between $\vb\astar$ and $\TS(\hat{\Ab})$. Define $\theta\astar_1$ as the angle between $\vb\astar_{proj}$ and $\TS(\hat{\Ab})$. By the triangle inequality, $\hat{\theta} \leq \theta + \theta\astar_1$. Via an analogous projection of $\hat{\vb}\astar_{proj}$ onto $\TS(\Ab)$ (blue dashed line in right panel of Figure~\ref{thetastars}), define $\theta\astar_2$ as the angle between $\hat{\vb}\astar_{proj}$ and $\TS(\Ab)$. Then $\theta$, the angle between $\vb\astar$ and $\TS(\Ab)$, is no larger than the total angle traversed by projecting $\vb\astar$ onto $\TS(\hat{\Ab})$ and then projecting that result, $\hat{\vb}\astar_{proj}$, onto $\TS(\Ab)$: $\theta \leq \hat{\theta} + \theta\astar_2$. 

\begin{figure}[h]
\includegraphics[scale=0.3]{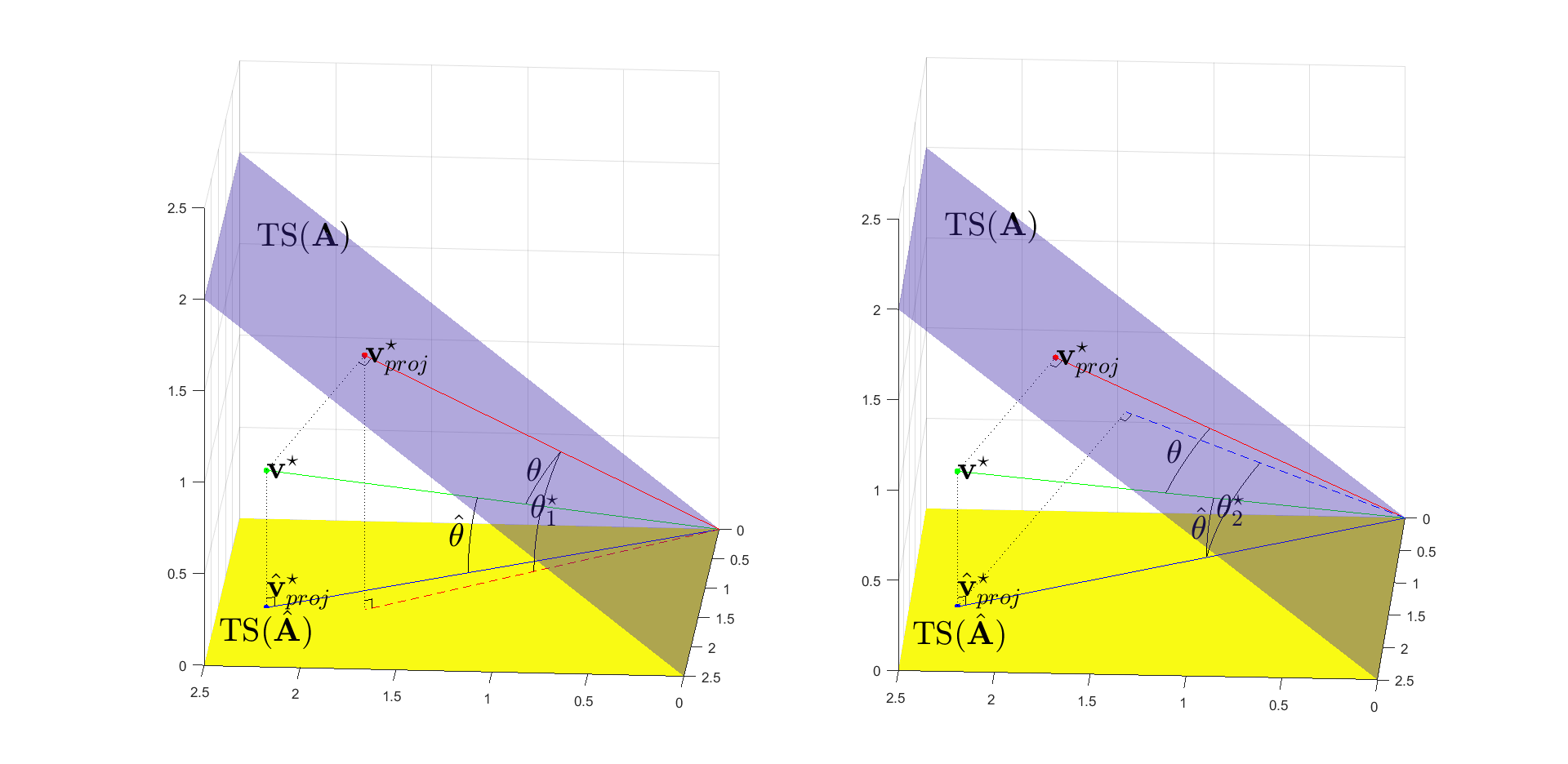}
\caption{Locations of $\theta$, $\hat{\theta}$, $\theta\astar_1$, and $\theta\astar_2$ in a low-dimensional example. Each panel demonstrates a different angle bound. Left: $\hat{\theta} \leq \theta + \theta\astar_1$. Right: $\theta \leq \hat{\theta} + \theta\astar_2$.}\label{thetastars}
\end{figure}

The above discussion of angles between subspaces summarizes the proof of the following theorem. More details can be found in the Ph.D. dissertation of \citet{mjdiss}.
\begin{theorem}
\label{thetastarMJ}
Let $\Xb = \Ab + \Eb$ be a $d\times n$ data matrix which is a sum of a signal matrix $\Ab$ and a noise matrix $\Eb$ under the assumptions of \eqref{eq:datamodel}. Given $\theta$, $\hat{\theta}$, $\theta_1\astar$, and $\theta_2\astar$ defined as above, and using $(\bullet)_+ = \max(\bullet, 0)$, we have:
\begin{equation}
\label{thetastarbounds}
(\hat{\theta} - \theta_1\astar)_+ \leq \theta \leq \hat{\theta}+\theta_2\astar.
\end{equation}
\end{theorem}
Both inequalities in \eqref{thetastarbounds} will be used to rule out directions $\vb\astar$ as a candidate for the joint spaces. These exclusions will be based on two distinct statistical arguments: a novel rotational bootstrap used for both bounds, and a distribution of angles between random directions used only for the upper bound.

If the angle to the estimated signal subspace $\hat{\theta}$ for a given candidate direction $\vb\astar$ is less than $\theta_1\astar$, then the lower bound in \eqref{thetastarbounds} for the angle to the true signal subspace is 0, indicating that this direction can't be ruled out as lying in the true signal subspace. The angle $\theta\astar_1$ is behaving much like the desired perturbation angle bound, but as noted in \citet{mjdiss}, $\theta_1\astar$ is not directly estimable for a given direction. However, $\theta_1\astar$ is uniformly bounded from above for all $\vb\astar$ by the maximum principal angle $\phi$ between $\TS(\Ab)$ and $\TS(\hat{\Ab})$ (see Appendix~\ref{sec:paa}). Our chosen perturbation angle bound will therefore be a statistical estimate $\hat{\phi}$ of that maximum principal angle $\phi$. This estimation is performed via a rotational bootstrap as described in Section~\ref{sec:rotboot}.

Unlike $\theta_1\astar$, the angle $\theta_2\astar$ in the upper bound of \eqref{thetastarbounds} can be estimated for each $\vb\astar$ using
\begin{align}
\begin{split}
\label{eq:thetatwostar}
\theta_2\astar 
&= \arccos\left(\frac{\parallel \Vb^{\top}\hat{\Vb}\hat{\Vb}^{\top}\vb\astar \parallel}{\parallel \hat{\Vb}\hat{\Vb}^{\top}\vb\astar \parallel}\right).
\end{split}
\end{align}
The only unknown quantity in this formula is the matrix $\Vb^{\top} \hat{\Vb}$, and we generate samples from the estimated distribution of this matrix as part of the rotational bootstrap in Section~\ref{sec:rotboot}. Therefore by recording those samples and using them in \eqref{eq:thetatwostar} we can generate from a bootstrap distribution of $\theta_2\astar$ and choose a high percentile denoted as $\hat{\theta}_2\astar$. 

The main use of this estimate is determining whether $\vb\astar$ can be distinguished from an arbitrarily chosen direction based on the estimated upper bound $\hat{\theta}+\hat{\theta}\astar_2$. The rotational invariance property assumed of the signal and noise in \eqref{eq:datamodel} implies a natural null distribution for comparison. In particular, we choose the distribution of angles between a fixed arbitrary $\hat{r}$-dimensional subspace of $\bbR^n$ (recall that $\hat{r}$ is the estimated signal rank) and unit vectors chosen uniformly at random. We pick a \textit{random direction angle bound} $\theta_0$ as a low percentile of that null distribution. If $\hat{\theta}+\hat{\theta}\astar_2$ lies above $\theta_0$ for some direction $\vb\astar$, then that direction cannot be distinguished from an arbitrarily chosen direction, which provides statistical evidence that $\vb\astar$ is far from $\TS(\Ab)$.



The above derivations of a perturbation angle bound and other angle-based inference have taken place entirely in trait space. Analogous derivations can be carried out in object space, and the estimation of perturbation angle bounds in both spaces can take place simultaneously during the rotational bootstrap. When considering candidate directions $\vb\astar$, we should additionally rule out directions whose corresponding basis directions in object space do not obey the object space perturbation angle bounds. Therefore both space's angle bounds play key roles in the optimization problem for locating joint structure between data blocks. This leads to more precise estimates of joint subspaces compared to methods like AJIVE \citep{ajive} that consider only trait space information in their algorithms.

\subsubsection{Rotational Bootstrap}\label{sec:rotboot}

We estimate perturbation angle bounds for the object space and trait space of a data block using a novel \textit{rotational bootstrap}. This technique is designed to take advantage of the assumed rotational invariance property and  aims to estimate the distribution of principal angles between object space and trait space subspaces of $\Xb$ and $\Ab$ through random generation of replicate signal subspaces.

Recall $\Xb=\Ab+\Eb$ from \eqref{eq:datamodel}, and as in Section~\ref{sec:mpsignal} the compact SVD of the rank $r$ signal matrix is $\Ab=\Ub \Db \Vb^\top$ and the compact SVD of the rank $r$ approximation to the data $\Xb$ is 
$\hat\Ab=\hat\Ub \hat\Db \hat\Vb^\top$. Next consider a random replication $\Xb^\diamond=\Ub^\diamond \Db \Vb^{\diamond\top}+\Eb^\diamond$, where $\Eb^\diamond$ has the same distribution as $\Eb$, and $\Ub^\diamond$ and $\Vb^\diamond$ are random $d\times r$ and $n\times r$ orthonormal matrices, respectively. The corresponding compact SVD of the rank $r$ approximation to $\Xb^\diamond$ is $\hat\Ab^\diamond=\hat\Ub^\diamond \hat\Db^\diamond \hat\Vb^{\diamond\top}$.
\begin{assumption}\label{as:rotation}
 The data matrix model \eqref{eq:datamodel} is called {\em rotationally invariant} when the matrices  $\hat \Db, \Vb^\top\hat\Vb, \hat\Vb^\top\Vb, \Ub^\top\hat\Ub, \hat\Ub^\top\Ub$ have the same distribution as the corresponding matrices \\ $\hat\Db^\diamond,\Vb^{\diamond\top}\hat\Vb^\diamond, \hat\Vb^{\diamond\top}\Vb^\diamond, \Ub^{\diamond\top}\hat\Ub^\diamond, \hat\Ub^{\diamond\top}\Ub^{\diamond}$.
\end{assumption}
As discussed in Appendix~\ref{sec:paa}, these matrices determine the principle angle structure between the spaces spanned by the low rank signal and its estimate in both trait and object spaces.
Theorem~7 of \citet{mjdiss}  shows that if the noise distribution is rotationally invariant, e.g., having i.i.d.~centered Gaussian entries, then the model satisfies Assumption~\ref{as:rotation}. Alternatively, the model will be rotationally invariant if the signal matrix $\Ab$ is considered random following a rotationally invariant prior distribution akin to (54) in \citet{gd2017}.

The continuity of these distributions in the singular values $\Db$ suggests use of a {\em parametric bootstrap} estimator of these quantities based on the estimated $\hat r\times\hat r$ singular value matrix $\hat\Db$.
In particular, we form a bootstrap replication of a signal matrix $\Ab^{\circ} = \Ub^{\circ} \hat{\Db} \Vb^{\circ\top}$, where $\Ub^{\circ}$ and $\Vb^{\circ}$  are random $d\times\hat{r}$ and $n\times\hat{r}$ orthonormal matrices, respectively. Using this randomly rotated estimated signal matrix along with an estimate $\hat{\Eb}$ of the noise matrix $\Eb$, we form a bootstrap replication of the data matrix $\Xb^{\circ} = \Ab^{\circ} + \hat{\Eb}$. If the rotational invariance assumption is satisfied, and the estimated $\hat\Db$ is close to $\Db$ this construction produces replicate signal and data matrices with principal angle structure drawn from a similar distribution as the unobserved principal angle structure between the true signal and data matrices. 
An important, and perhaps surprising point is that the na\"ive noise matrix estimate $\hat{\Eb} = \Xb - \hat{\Ab}$ is not appropriate for use in this construction. This is because $\Xb-\hat{\Ab}$ has insufficient energy in the directions associated with $\hat{\Ab}$, and therefore has eigenvalues that don't follow the Marchenko-Pastur distribution in the manner expected for a noise matrix under our assumptions. Our proposed estimator, labeled $\hat{\Eb}_{impute}$, is shown in \eqref{eq:ehatimpute} and corrects for the insufficient energy through imputation via Marchenko-Pastur random variates. See Appendix~\ref{sec:mp} for details on the Marchenko-Pastur distribution and Appendix~\ref{sec:noise} for full details on the poor performance of $\Xb-\hat{\Ab}$ and the motivation of $\hat{\Eb}_{impute}$.


As mentioned in Section~\ref{sec:perturb}, we will use estimates of the maximum principal angles between the subspaces spanned in object space and trait space by $\Xb$ and $\Ab$ as perturbation angle bounds. Through repeated replications of the randomly rotated signal and data matrices described in the previous paragraph, we generate bootstrap samples estimating the \textit{distribution} of principal angles between subspaces spanned by $\Xb$ and $\Ab$ in both object and trait space. With sufficiently many replications (we use $M = 400$) we can choose high quantiles (e.g. 0.95) of the empirical distributions of maximum principal angles as statistical perturbation angle bounds. Recall that $\hat{\phi}$ is the trait space perturbation angle bound, and denote the corresponding object space perturbation angle bound as $\hat{\psi}$. 

The procedure described in Section~\ref{sec:mpsignal} discriminates noise fairly well, but as our algorithm is based on angles we find that additional angle-based rank selection is often necessary for good practical performance. Therefore as part of the rotational bootstrap we \textit{filter} the signal subspaces according to the random direction angle bound $\theta_0$ defined in Section~\ref{sec:perturb}. In particular we choose a filtered rank $\check{r}$ such that the estimated maximum principal angles between true and estimated signal don't exceed $\xi\theta_0$, where $\xi\in(0,0.5]$ is a tuning parameter. In our analyses we explored a grid of values for $\xi$ ranging from 0.3 to 0.5 and found that a value between 0.35 and 0.4 often captured an appropriate amount of signal for our data sets. Therefore the case studies in Sections~\ref{sec:genomics} and \ref{sec:data} choose $\xi=1-\frac{2}{1+\sqrt{5}}\approx 0.382$, a value based on the golden ratio. This hyperparameter can be tuned up or down within $(0, 0.5]$ to include more or less information from the estimated signal matrix in the analysis.


The limitation of $\xi\leq 0.5$ follows from the statistical inference framework laid out in Section~\ref{sec:perturb}. If the lower bound $(\hat{\theta} - \hat{\phi})_+$ is 0 and the upper bound $\hat{\theta} + \hat{\theta}_2\astar$ is simultaneously greater than $\theta_0$ for a given candidate direction $\vb\astar$, the inference procedure says there is evidence that $\vb\astar$ is both significantly close to the true signal subspace and indistinguishable from an arbitrary direction. This inference outcome is completely non-informative. In this case, both $\hat{\theta}$ and $\hat{\theta}_2\astar$ must be bounded from above by the maximum principal angle between estimated and true signal subspaces. Therefore this non-informative inference outcome is avoided by filtering the estimated signal subspace until the rotational-bootstrap-estimated maximum principal angle is at most $0.5\cdot\theta_0$.




The above description of the rotational bootstrap algorithm is formulated in Algorithm~\ref{alg:rotboot}. For each block $k$, we have as inputs to the algorithm the estimated signal matrix $\hat{\Ab} = \hat{\Ub}\hat{\Db}\hat{\Vb}^\top$ from \eqref{eq:finala} in Section~\ref{sec:sse}, the estimated residual matrix $\hat{\Eb}_{impute}$ from \eqref{eq:ehatimpute} in Appendix~\ref{sec:noise}, the random direction angle bound $\theta_0$ discussed at the end of Section~\ref{sec:perturb}, and the hyperparameters $\xi$ and $\alpha$. $\alpha$ is the desired confidence level for the perturbation angle bounds and the default is 0.95. At the end of the algorithm, we have as outputs estimates of the trait space and object space perturbation angle bounds $\hat{\psi}$ and $\hat{\phi}$ respectively for each block $k$.

During each replication, random subspaces are generated from i.i.d. standard Gaussian matrices with the same centering operations used on the data. Note that orthogonalization of an i.i.d. random matrix in this fashion is identical to sampling from a rotationally uniform distribution of subspaces, according to Theorem 2.2.1 from \citet{chikuse2012statistics}. The inner \textbf{for} loop records maximum principal angles at each possible filtered rank from 1 to $\hat{r}$. After the outer \textbf{for} loop concludes, the algorithm chooses a filtered rank $\check{r}$ to align with the chosen value of $\xi$. In the case where the filtered ranks in object space and trait space differ, the smaller of the two is chosen to ensure compliance in both spaces. The trait and object space perturbation angle bounds $\hat{\phi}$ and $\hat{\psi}$ are chosen as the $1-\alpha$ percentile of the empirical distributions of angles at the filtered rank $\check{r}$. Once the filtered rank is selected, 
we filter the columns of the estimated basis matrices for the signal object and trait space subspaces to correspond with the reduced rank. Let $\check{\Ub} = \overbar{\Ub}_{1:\check{r}}$ and $\check{\Vb} = \overbar{\Vb}_{1:\check{r}}$ be the final estimates of the signal object and trait space bases respectively.

\begin{algorithm}[H]
\caption{Rotational Bootstrap}\label{alg:rotboot}
\begin{algorithmic}
\Require $\hat{\Db}$: $\hat{r}\times\hat{r}$ diagonal singular values matrix, $\hat{\Eb}$: $d\times n$ residual matrix, $\theta_0$: random direction angle bound, $\xi$: filter percentage, $\alpha$: significance level, $M$: number of replications
\State $objectAngles\gets 90*\1^{M\times \hat{r}}$; $traitAngles\gets 90*\1^{M\times \hat{r}}$
\ForAll{$m\in \{1,\dots,M\}$}
	\State $\Ub^{\circ}\gets \text{rand}^{d\times\hat{r}}$; $\Vb^{\circ}\gets \text{rand}^{n\times\hat{r}}$
	\State \Comment Replications must be orthogonal to constant function direction in appropriate spaces.
	\If{$\Xb$ is trait-centered}
		\State $\Ub^{\circ}\gets \left(\Ib^{d\times d} - \frac{1}{d}\1^{d\times d}\right) \Ub^{\circ}$
	\EndIf
	\If{$\Xb$ is object-centered}
		\State $\Vb^{\circ}\gets \left(\Ib^{n\times n} - \frac{1}{n}\1^{n\times n}\right) \Vb^{\circ}$
	\EndIf
	\State $\Ub^{\circ} \gets \text{orth}(\Ub^{\circ})$; $\Vb^{\circ} \gets \text{orth}(\Vb^{\circ})$; $\Ab^{\circ} \gets \Ub^{\circ} \hat{\Db}\Vb^{\circ}$; $\Xb^{\circ} \gets \Ab^{\circ} + \hat{\Eb}$
	\State $[\overbar{\Ub}^{\circ},\sim,\overbar{\Vb}^{\circ}]\gets \text{SVD}(\Xb^{\circ})$. ``MATLAB" Notation for only storing the orthonormal matrices
	\ForAll{$j\in \{1,\dots,\hat{r}\}$}
		\State \Comment Smallest singular value equal to cosine of largest angle.
		\State $[\sim,\vec{\nu}_{object} ,\sim] \gets \text{SVD}\left(\Ub^{\circ T} \overbar{\Ub}^{\circ}_{1:j}\right)$; $[\sim,\vec{\nu}_{trait} ,\sim] \gets \text{SVD}\left(\Vb^{\circ T} \overbar{\Vb}^{\circ}_{1:j}\right)$
		\State $objectAngles[m,j] \gets\arccos(\min(\vec{\nu}_{object}))$; $traitAngles[m,j] \gets\arccos(\min(\vec{\nu}_{trait}))$
	\EndFor	
\EndFor
\State $objectAnglesSort\gets \text{sort}(objectAngles, col, asc);traitAnglesSort\gets \text{sort}(traitAngles, col, asc)$ 
\State $\check{r} \gets \min\left(\sum_{j=1}^{\hat{r}} \mathbb{I}_{\{objectAnglesSort[\alpha M,j] < \xi\theta_0\}}, \sum_{j=1}^{\hat{r}} \mathbb{I}_{\{traitAnglesSort[\alpha M,j] < \xi\theta_0\}} \right)$
\State $\hat{\psi} \gets objectAnglesSort[\alpha M,\check{r}]$; $\hat{\phi} \gets traitAnglesSort[\alpha M,\check{r}]$ 
\end{algorithmic}
\end{algorithm}

\subsubsection{Signal Extraction for Synthetic Data}\label{sec:syntheticsse}

The results of signal space extraction on the synthetic data example from Figure~\ref{3blocktruth} are shown in Figure~\ref{toyexamplestep1}. Each heatmap shows the estimated signal matrix $\hat{\Ab}$ for the respective data block. The denoising of each data block appears to have been successful when comparing the visual impression of these heatmaps to the original data matrices shown in the first column of Figure~\ref{3blocktruth}. The signal rank is correctly chosen as 3 for all three blocks. In this case, with a generally high signal-to-noise ratio, angle-based rank filtering didn't further reduce the rank beyond the value determined from eigenvalue-based rank selection. 

\begin{figure}[h]
        \begin{minipage}[c]{0.3\linewidth}
            \centering
            \includegraphics[width=\textwidth]{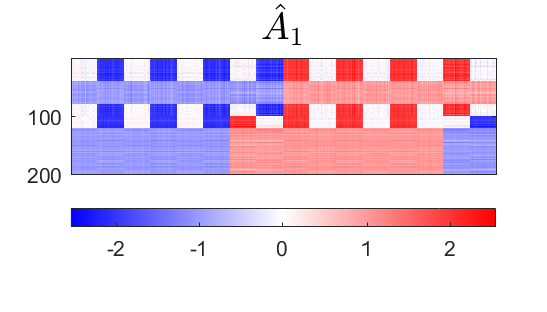}
        \end{minipage}
        \begin{minipage}[c]{0.3\linewidth}
            \centering
            \includegraphics[width=\textwidth]{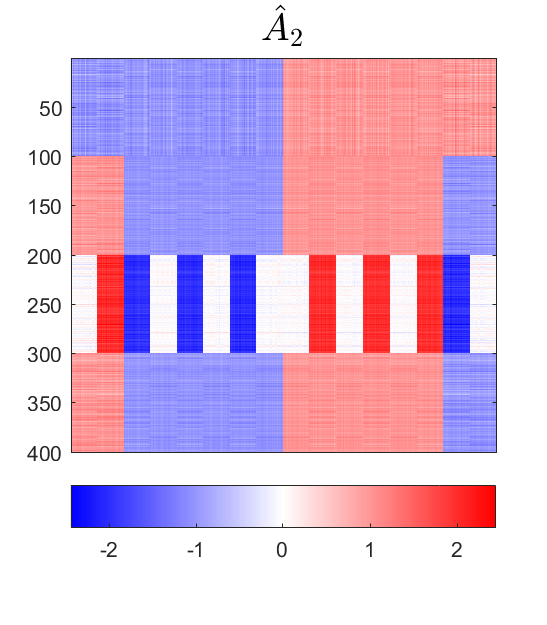}
        \end{minipage}
        \begin{minipage}[c]{0.3\linewidth}
            \centering
            \includegraphics[width=\textwidth]{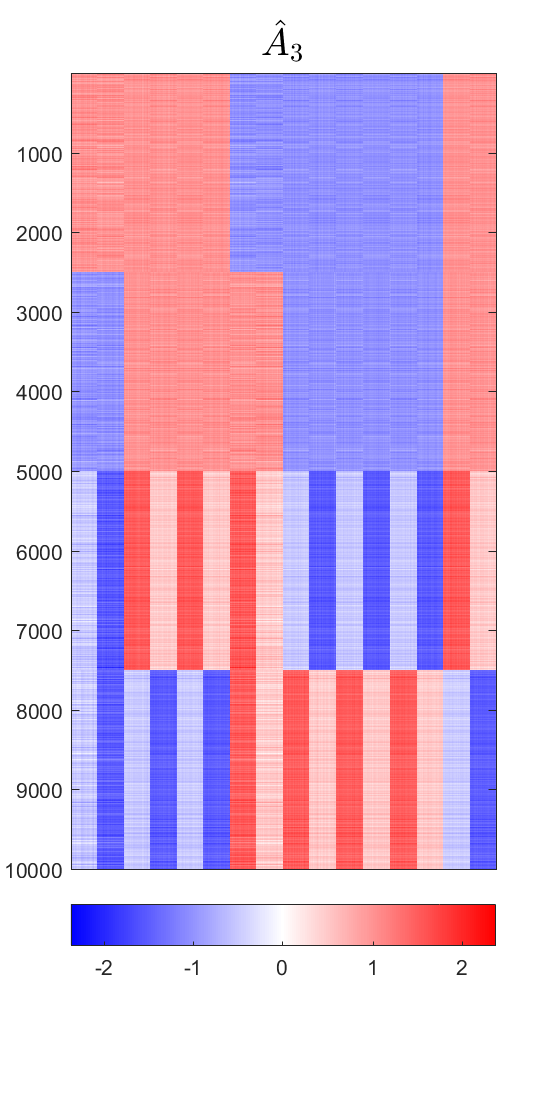}
        \end{minipage}
        \caption{Estimated signal matrices for each block in the synthetic example defined in Figure \ref{3blocktruth}. The heatmaps show good recovery of the original signal patterns. The trait and object spaces of each estimated matrix are rank 3.}\label{toyexamplestep1}
\end{figure}

Since we know the true signal object and trait subspaces for the synthetic data example, we can compare the estimated perturbation angle bounds to the actual angles between estimated and true signal subspaces to check the performance of the bounds. Tables~\ref{tbl:traittruthangles} and \ref{tbl:objecttruthangles} display the perturbation angle bounds and angles between the true and estimated signal subspaces in trait space and object space respectively. The calculated bounds exceed the true angles in all cases, so the true basis directions all lie within the cones of feasibility defined by the bounds. Since the bounds are calculated as uniform 95\% bounds, we'd expect to not cover the truth about 1 in 20 times, and the performance of the bounds in this case aligns with that expectation. 

Each synthetic data block has a similar signal-to-noise ratio, so the observed differences in perturbation angle bounds in the second column of each table are primarily explained by differences in matrix dimension. Recall that all four data blocks have 400 data objects, $X_1$ has 200 traits, $X_2$ has 400 traits, and $X_3$ has 10000 traits. The trait space perturbation angle bounds decrease as the number of traits in the data blocks increase since we have a more precise idea of where the true signal subspace is with more trait vectors in the same trait space $\bbR^{400}$. The object space perturbation angle bounds increase as the number of traits in the data blocks increase since we have a more precise idea of where the true signal subspace is in object space with 400 object vectors in $\bbR^{200}$ than we do with 400 object vectors in $\bbR^{10000}$.

\begin{table}[t]
\centering
\begin{tabular}{cccccc}\hline
\thead{Data Block} & \thead{Trait Space \\ Angle Bound} & \thead{Angle to \\ 1,2,3 Truth} & \thead{Angle to \\ 1,2 Truth} & \thead{Angle to \\1,3 Truth} & \thead{Angle to \\ 2,3 Truth} \\ \hline
1 & 11.7 & 9.2 & 8.5 & 6.1 & \\
2 & 8.6 & 6.9 & 5.6 & & 4.0 \\
3 & 2.8 & 2.5 & & 1.0 & 1.0 \\ \hline
\end{tabular}
\caption{Table of angles between estimated signal trait spaces and true signal trait spaces. All angles are within the calculated perturbation angle bounds.} \label{tbl:traittruthangles}
\end{table}

\begin{table}[t]
\centering
\begin{tabular}{cccccc}\hline
\thead{Data Block} & \thead{Object Space \\ Angle Bound} & \thead{Angle to \\ 1,2,3 Truth} & \thead{Angle to \\ 1,2 Truth} & \thead{Angle to \\1,3 Truth} & \thead{Angle to \\ 2,3 Truth} \\ \hline
1 & 8.6 & 4.5 & 4.9 & 4.7 &\\
2 & 8.6 & 5.8 & 6.6 & & 4.0\\
3 & 13.1 & 7.9 & & 4.6 & 4.7\\ \hline
\end{tabular}
\caption{Table of angles between estimated signal object spaces and true signal object spaces. All angles are within the calculated perturbation angle bounds.} \label{tbl:objecttruthangles}
\end{table}

\subsection{Joint Subspace Estimation}\label{sec:jse}

We formally introduce the optimization problem for locating shared structure in DIVAS. 
The conceptual constraints and objective function are shown in \eqref{DivasOptHeur} and the full numerical algorithm is deferred to Appendix~\ref{sec:optdetails} with the main subproblem being a convex optimization problem \eqref{DivasOptCCSlackLoad}. 

For any given collection of blocks $\ib$, the corresponding joint subspace should be near each of the included blocks in some sense. In DIVAS, proximity is evaluated in terms of angles between candidate directions and subspaces. In particular, during each phase of joint subspace estimation we minimize the angle between candidate directions $\vb\astar$ and the estimated trait space subspaces of included blocks subject to identifiability and feasibility constraints:
\begin{equation}\label{DivasOptHeur}
\begin{aligned} 
&\min_{\vb\astar}& & - \sum_{k\in \ib} \cos^2 \hat{\theta}_{Tk}  & & \\
&s.t.&  \hat{\theta}_{Tk} &= \angle(\vb\astar,\check{\Vb}_k) & \forall k & \\
&  & \hat{\theta}_{Ok} &= \angle(\Xb_k\vb\astar,\check{\Ub}_k) & \forall k & \\
&  &\hat{\theta}_{Tk} &\leq \hat{\phi}_k & \forall k &\in \ib \\
&  &\hat{\theta}_{Tk} &> \hat{\phi}_k & \forall k &\in \ib^c \\
&  &\hat{\theta}_{Ok} &\leq \hat{\psi}_k & \forall k &\in \ib \\
&  &\vb\astar &\perp \mathfrak{V}_{\jb} & \forall \jb &\supseteq \ib. \\
\end{aligned}
\end{equation}
In practice, this problem is solved via an iterative procedure called convex-concave procedure as described in \citep{il2016}. 
This algorithm is also called DC (Difference of two Convex functions) algorithm in the literature.
Full explanation can be found in Appendix~\ref{sec:optdetails}. In the rest of this section we discuss details of the optimization problem \eqref{DivasOptHeur} and its role in the estimation of the joint signal structure.

The objective function is expressed in terms of angle cosines in line 1 of \eqref{DivasOptHeur}. To ensure that a candidate direction lies in the true signal subspace of an included block $\Xb_k, k\in \ib$ with high significance, the trait space angle between a candidate direction and the subspace spanned by the columns of $\check{\Vb}_k$ should be at most the trait space angle perturbation bound $\hat{\phi}_k$. Additionally, the object space angle between $\Xb_k\vb\astar$ and the subspace spanned by the columns of $\check{\Ub}_k$ should be at most the object space angle perturbation bound $\hat{\psi}_k$. Finally, the angle between a candidate direction and an excluded block should be at least the trait space angle perturbation bound $\hat{\phi}_k$. These requirements are expressed as constraints for the optimization problem in lines 2-6 in \eqref{DivasOptHeur}, with subscripts $T$ and $O$ indicating angles in trait space and object space respectively. Crucially, our dimensionally flexible subspace-based angle perturbation approach to signal extraction allows object space information to be incorporated very naturally into the joint subspace estimation algorithm. This innovation enhances the significance and interpretability of loadings vectors found using DIVAS.


Following the proof of Theorem~\ref{th:construct}, we determine each block collection's potential joint structure in turn, starting with larger block collections and ending with singleton block collections. Within a joint structure search for a given block collection $\ib$, joint subspace basis directions are found one at a time via successive solves of \eqref{DivasOptHeur}. If no new feasible direction is found, the search among the current block collection ends and the search among the next block collection begins. Candidate directions in trait space for a particular block collection must also obey orthogonality constraints expressed in parts 1 and 2 of Condition~\ref{asump:decomposition}. These conditions are concisely expressed in the constraint in line 7 of \eqref{DivasOptHeur}. The Gothic script symbol $\mathfrak{V}_{\ib}$ is used to denote the current estimated trait space basis for the shared structure for block collection $\ib$, i.e., an estimate of the $\Vb_\ib$ from \eqref{eq:signalmodel}. Note that the search for joint structure between two block collections of the same size is embarassingly parallelizable, as the orthogonality constraint will only include joint structure found for strictly larger block collections. The order in which block collections of equal size are searched does not affect DIVAS output.

As the algorithm finds basis directions for block collection $\ib$, the rank of $\mathfrak{V}_{\ib}$ increases, and so the constraint in line 7 tightens as more basis directions are located. As directions are located the angle constraints also change. We shrink $\check{\Vb}_k$, the orthonormal bases for estimated trait space signal, to only include directions in the null space of $[\mathfrak{V}_{\jb}]_{\jb\supseteq\ib}$. This basis shrinking improves computation time and assists in the choice of basis directions satisfying our assumptions.

\subsection{Signal Reconstruction}\label{sec:recon}

 
Once we have located all possible joint structure, the remaining task is to reconstruct the signal matrix components for each data block. Recall that in Section~\ref{sec:jse} we denote the estimated orthonormal basis for the joint structure among blocks in collection $\ib$ as $\mathfrak{V}_{\ib}$. For a given data block $k$, we first horizontally concatenate all joint structure basis matrices found involving block $k$ into one matrix $\left[\mathfrak{V}_{\ib}\right]_{\ib|k\in \ib}$. Then we form a linear regression problem to find the corresponding loadings vectors for block $k$ associated with the common scores vectors for block collection $\ib$ in a similar fashion as the proof of Theorem~\ref{th:construct}. 
In particular  $[\mathfrak{L}_{\ib,k}]_{\ib|k\in\ib}$ is chosen as the least-squares solution of the regression problem $\min_{\mathfrak{L}}\|\Xb_k - \mathfrak{L}\cdot [\mathfrak{V}_{\ib}]^{\top}_{\ib|k\in \ib}\|_2^2$. This solution is unique when $[\mathfrak{V}_{\ib}]_{\ib|k\in \ib}$  is full-rank. The columns of $[\mathfrak{L}_{\ib,k}]_{\ib|k\in\ib}$ can then be partitioned into loadings matrices $\mathfrak{L}_{\ib,k}$, each of which is associated with joint structure for one block collection $\ib$ with $k\in\ib$, i.e., $\mathfrak{L}_{\ib,k}$ is an estimate of $\Lb_{\ib,k}$ from \eqref{eq:signalmodel}.  The estimated partially shared joint structure between between data blocks in $\ib$ is then simply $\hat{\Ab}_{\ib,k}=\mathfrak{L}_{\ib,k} \mathfrak{V}_{\ib}^{\top}$. 

As a brief remark, DIVAS may in certain cases select shared structure such that the rank of $\left[\mathfrak{V}_{\ib}\right]_{\ib|k\in \ib}$ is larger than $\check{r}_k$. Since the subspaces spanned by $\mathfrak{V}_{\ib}$ and $\mathfrak{V}_{\jb}$ need not be orthogonal unless $\ib \subseteq \jb$ or $\jb \subseteq \ib$, more than $\check{r}_k$ total basis directions may be selected within the cone of feasibility for block $k$. 

Additional insight comes from further decomposition of $\hat{\Ab}_{\ib,k}$ into a sum of rank 1 modes of variation. Each mode comes from an outer product of corresponding columns of $\mathfrak{L}_{\ib,k}$ and $\mathfrak{V}_{\ib}$. Since $\mathfrak{L}_{\ib,k}$ and $\mathfrak{V}_{\ib}$ are only determined up to basis rotation, we first select a rotation matrix $\Qb_{\ib}$, and then examine modes of variation formed with the matrices $\mathfrak{L}_{\ib,k}\Qb_{\ib}$ and $\mathfrak{V}_{\ib}\Qb_{\ib}$. In particular, we take an SVD of the projection of the stacked data matrix $[\Xb_k^{\top}]^{\top}_{k\in\ib}$ onto the subspace spanned by $\mathfrak{V}_{\ib}$ in trait space, and choose $\Qb_{\ib}$ as the matrix of right singular vectors from that calculation. This re-rotation can be thought of as sorting the modes of variation within the shared subspace in order of importance.

Since DIVAS is based on angles, additional diagnostic insight into the modes of variation is derived from angles between the loadings (columns of $\mathfrak{L}_{\ib,k}\Qb_{\ib})$ and scores (columns of $\mathfrak{V}_{\ib}\Qb_{\ib}$) and the object and trait spaces spanned by the estimated low rank matrix $\hat\Ab_k$, respectively. In particular, for each estimated score vector and loadings vector, the angle to each estimated signal matrix $\hat{\theta}_k$ and the upper bound on the angle to the true signal $\hat{\theta}_k+\theta_2\astar$ for each direction in both trait space and object space are computed. See Section~\ref{sec:perturb} for a definition of $\theta_2\astar$. To calculate the upper bound for one of the vectors, we choose the 95th percentile of an empirical distribution of $\theta_2\astar$ generated using \eqref{eq:thetatwostar} and the cached matrices from the rotational bootstrap (See Section~\ref{sec:rotboot}). Scores for all blocks are in a shared trait space, and therefore in trait space we calculate these angles not only for included ($k\in\ib$) but also for excluded ($k\notin\ib$) blocks for each block collection $\ib$. The angle to the included block is expected to be small and the angle to the excluded block is expected to be large, though not necessarily 90 degrees. If some score vector has an upper bound below the random direction bound $\theta_0$ for an excluded block, then the corresponding mode of variation is correlated with that excluded block even though it is not joint with that block. The object spaces are block specific, and therefore for loadings we calculate angles to the included ($k\in\ib$) data blocks only.  These diagnostic angles form the crux of the overall diagnostic displays that we describe next.

\subsection{DIVAS Diagnostic Graphics}\label{sec:graphics}

We compile all angle-based diagnostics for DIVAS into comprehensive displays. These displays can be seen in Figures~\ref{fig:rbbToyExample} and \ref{fig:rbbToyExampleLoad} for the synthetic data set shown in Figure~\ref{3blocktruth}, in Figure~\ref{fig:tcgascoresdiag} for breast cancer omics data, and in Figures~\ref{fig:mortscoresdiag} and \ref{fig:mortloadingsdiag} for 20th century mortality data. We explain the interpretation of these displays in this section using the synthetic data example.

Figure~\ref{fig:rbbToyExample} shows the diagnostic angles for the joint scores vectors found for the synthetic data example from Figure~\ref{3blocktruth}, and Figure~\ref{fig:rbbToyExampleLoad} shows those same diagnostic angles for the joint loadings vectors. Each row of boxes corresponds to a data block and the various block collections appear in the columns. Boxes for included blocks in a given column are colored-in while boxes for excluded blocks are white. The number in each colored box specifies the rank of the estimated joint subspace between the blocks included in that column. Block collections where no partially shared joint structure was found are labeled with a 0 and grayed out. The last column labeled \textit{Ranks} contains the dimensions of key subspaces for each data block $k=1,\dots , K$. The \textit{final} rank is the dimension of the subspace spanned by all structure involving that data block, i.e. the rank of $\left[\mathfrak{V}_{\ib}\right]_{\ib|k\in \ib}$. The \textit{filtered} rank is the dimension of the estimated signal subspace in both object and trait space for that data block, i.e. $\check{r}_k$. The \textit{maximum} rank is the largest possible dimension spanned by structure involving that data block, i.e. $d_k\wedge n$. These three ranks will usually appear in ascending order, but as discussed in Section~\ref{sec:recon}, the final rank is sometimes larger than the filtered rank.

To explain the interpretation of the comprehensive information in DIVAS diagnostic displays, we first focus on the top-left corner of Figure~\ref{fig:rbbToyExample}. Each box is a scatter plot, with the horizontal axis indicating basis direction index and the vertical axis indicating angle from $0^{\circ}$ at the bottom to $90^{\circ}$ at the top. Within a box of this figure, each candidate direction found for that column's joint structure is represented by two points: $\times$ and $\bullet$. The $\times$ represents the angle $\hat{\theta}_k$ between the candidate direction and the corresponding estimated signal matrix for data block $k$, and the $\bullet$ represents the upper bound $\hat{\theta}_k + \hat{\theta}_2\astar$ on $\theta_k$, the angle between the direction and the true subspace. The dashed line represents the perturbation angle bound $\hat{\phi}_k$ (for trait space) or $\hat{\psi}_k$ (for object space) and the dot-dash line represents the random direction angle bound $\theta_{0,k}$. The numerical values of those angle bounds are given to the right of each group of columns. As per the inferential framework laid out in Section~\ref{sec:perturb}, a $\times$ below the dashed line indicates strong evidence that the direction can't be ruled out as joint structure for that data block, and a $\bullet$ above the dot-dash line indicates strong statistical evidence that the direction can't be distinguished from an arbitrarily chosen direction with respect to that data block. Due to the rank filtering procedure that takes place during the rotational bootstrap, no direction has both a $\times$ below the dashed line and a $\bullet$ above the dot-dash line.

Based on the placements of $\times$ and $\bullet$ in each colored box in Figures~\ref{fig:rbbToyExample} and \ref{fig:rbbToyExampleLoad}, we have strong evidence that each piece of estimated joint structure located by DIVAS is statistically significant in both trait space and object space, respectively. Specifically, all $\times$ are below the perturbation angle bound dashed line within their respective boxes. We also gain additional insight about the angular relationships between the two-way joint subspaces via the angles to the excluded blocks in Figure~\ref{fig:rbbToyExample}. Each $\bullet$ lies below the random direction angle bound dot-dash line in columns 3-5 of the display, indicating strong evidence that the chosen joint subspaces are distinguishable from arbitrary directions with respect to the excluded block in each column. In fact, the true joint subspaces were constructed to have pairwise principal angles of $60^{\circ}$ between them, so this statistical rejection of arbitrariness is not surprising.

\begin{figure}[h]
\includegraphics[width=\textwidth]{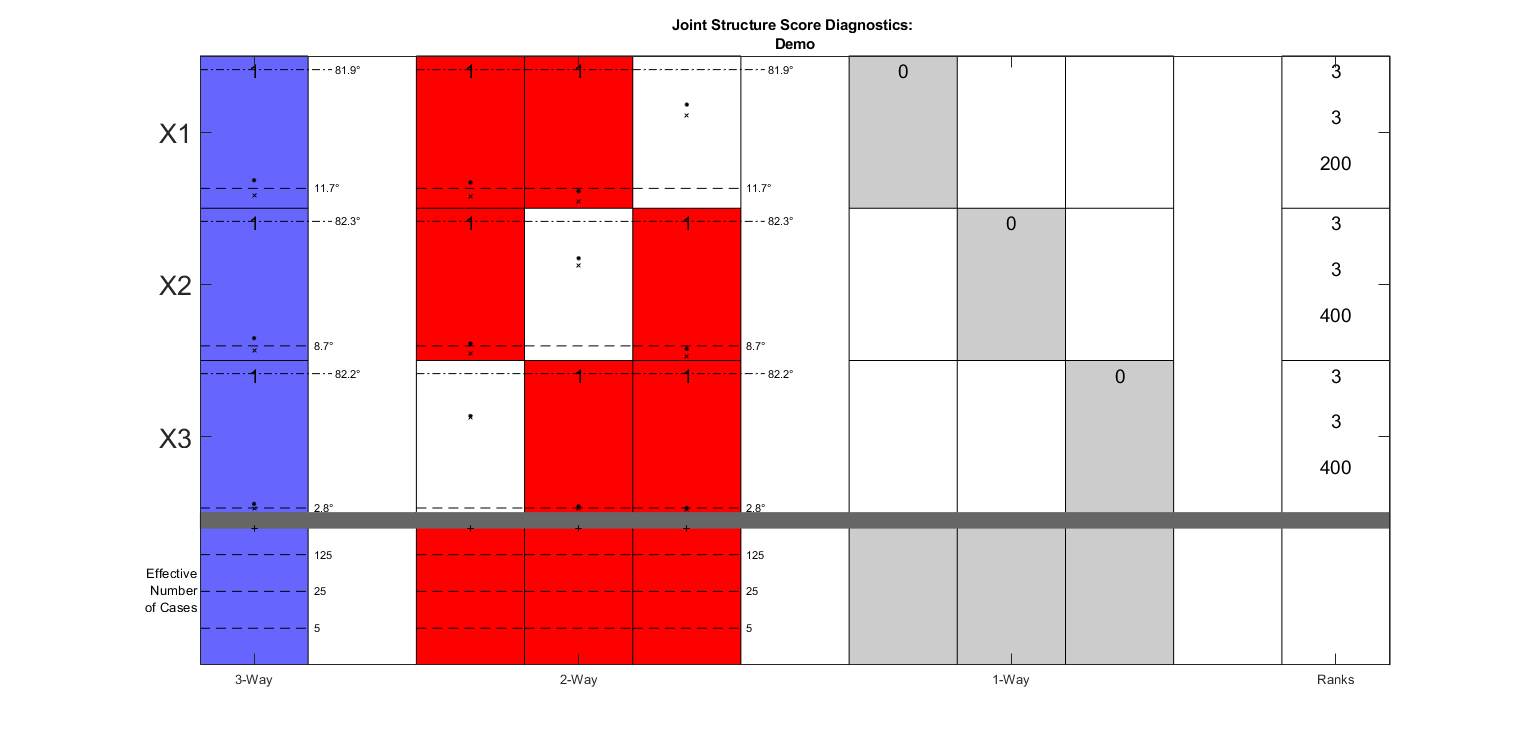}
\caption{Summary of joint structure diagnostics for the trait spaces of the synthetic data example. All joint structure located is statistically significant, and angles to excluded blocks confirm underlying angular relationships between two-way shared subspaces. Effective Numbers of Cases (ENC) values in last row align well with the true score vectors.}\label{fig:rbbToyExample}
\end{figure}

\begin{figure}[h]
\includegraphics[width=\textwidth]{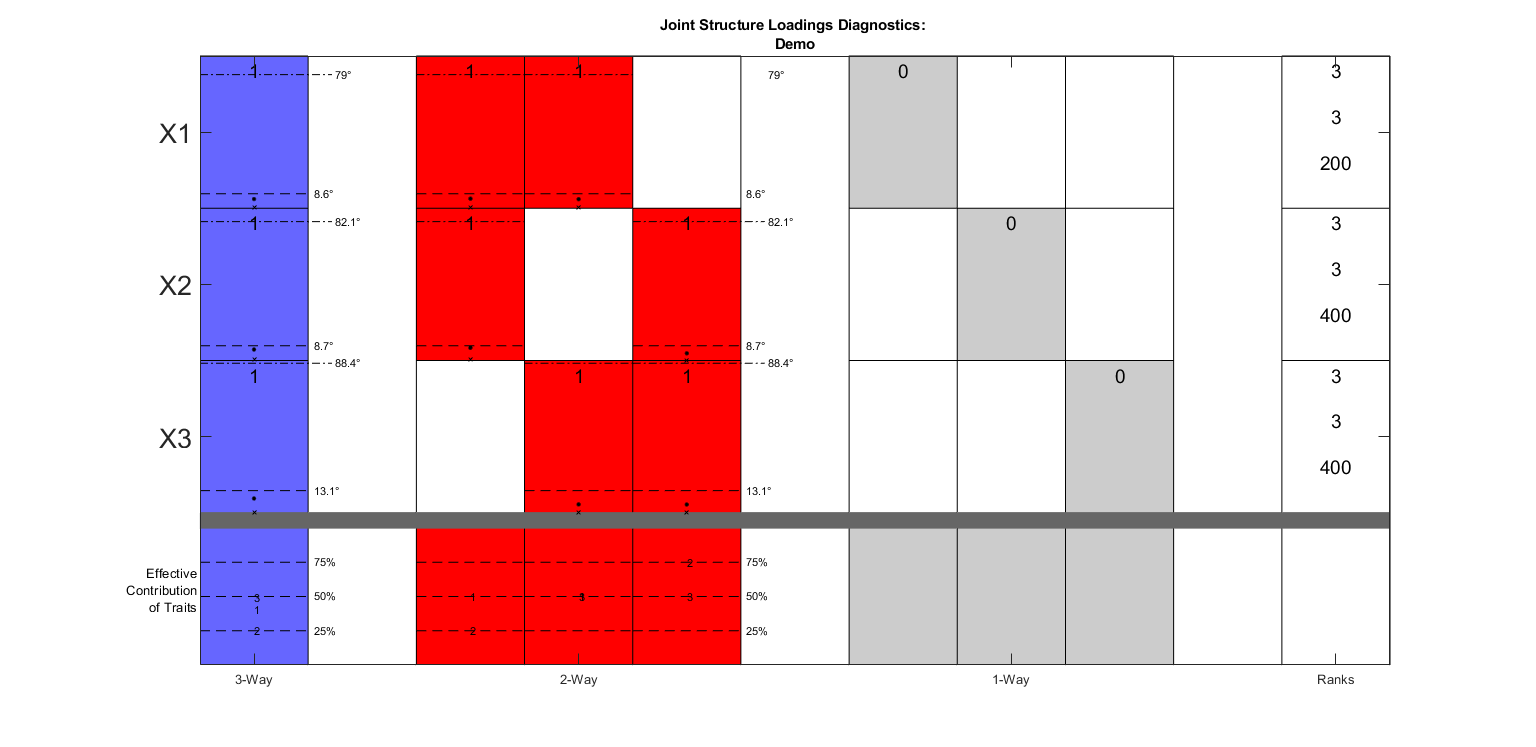}
\caption{Summary of joint structure diagnostics for the object spaces of the synthetic example. All joint structure located is statistically significant. Effective Contributions of Traits (ECT) in last row align with expectations per the proportion of colored rows in each heatmap of Figure~\ref{3blocktruth}.}\label{fig:rbbToyExampleLoad}
\end{figure}

There are some situations where DIVAS basis directions appear statistically significant from an angular perspective but depend on a very small number of observations or traits. Useful insight comes from summarizing the contributions of each observation and/or trait to the shared structure. In the case of observations, we quantify their involvement with a summary statistic called the \textit{Effective Number of Cases} (ENC) based on ideas in importance sampling from \citep{kish}. Let $v_j$ for $j\in\{1,\dots, n\}$ be the entries of a chosen direction $\vb\astar$. Note that the entries are scaled so $\vb\astar$ has norm 1 (i.e. $\sum_{j=1}^n v_j^2 = 1$). In this case, the ENC is:
\begin{equation*}
ENC = \frac{1}{\sum_{j=1}^n v_j^4}
\end{equation*}
If one entry $v_j$ is $\pm 1$ while the rest are 0, meaning a single observation determines the direction, then the ENC evaluates to 1. If all entries $v_j$ have the same magnitude $\pm\frac{1}{\sqrt{n}}$, meaning all observations have equal influence on the direction, then the ENC evaluates to $n$. Any chosen direction will fall somewhere between those two extremes.

The last row of Figure~\ref{fig:rbbToyExample} shows the ENC for each joint scores direction found for the synthetic data example. Each box is again a scatter plot, with the horizontal axis indicating basis direction index and the vertical axis indicating the ENC from 1 to $n$ on a logarithmic scale. Each ENC value is shown with a $+$. All the values for the synthetic example are very close to $n=400$, indicating near equal contribution from each observation in all the scores vectors. This aligns with expectations since the entries of the true shared scores directions all have equal magnitude.

In the case of summarizing individual trait contributions, we use an analogous metric that takes into account the differing magnitudes of loadings vectors within data blocks and the different dimensions of loadings vectors between data blocks. To differentiate the two metrics we call this one \textit{Effective Contribution of Traits} (ECT). ECT performs the same operation as ENC, except it uses the entries $l_m$ for $m\in\{1,\dots, d_k\}$  of candidate loadings directions $\lb\astar_k$. Furthermore, it scales the result by both the magnitude $||\lb_k\astar|| = \sum_{m=1}^{d_k} l_m^2$ of the candidate direction and by the number of traits $d_k$ to allow for comparisons between data blocks:
\begin{equation*}
ECT = \frac 1{d_k} \frac{\left(\sum_{m=1}^{d_k} l_m^2\right)^2}{\sum_{m=1}^{d_k} l_m^4}.
\end{equation*}

The last row of Figure~\ref{fig:rbbToyExampleLoad} shows the ECT for each loadings direction found for the synthetic data example. Within each box, the horizontal axis indicates basis direction index and the vertical axis indicates ECT percentage ranging from $0\%$ to $100\%$. Each block has its own loadings direction for each basis direction, so each block's ECT is shown with a number corresponding with that block's index in the analysis. In all cases, the contribution percentages align quite well with the percentage of traits involved in each piece of true joint structure as per Figure~\ref{3blocktruth}. For example, half of the traits in X3 have the characteristic pinstripe pattern of the fully joint structure in Figure~\ref{3blocktruth}, and the ``3" in the bottom-left box of Figure~\ref{fig:rbbToyExampleLoad} sits right around $50\%$.

\section{Case Studies}

\subsection{Case Study 1: Cancer Genomics}\label{sec:genomics}

One of the examples that motivates the development of DIVAS is a four-block data set containing different views of omics data from $n=616$ breast cancer patients from The Cancer Genome Atlas (TCGA) \citep{tcga}. We have a gene expression (GE) data block containing 16615 gene traits, a gene copy number (CN) data block with 24174 traits, a protein expression (RPPA) data block containing 187 protein traits, and a 0-1 mutation detection (Mut) block containing traits for 128 genes. Each patient is labeled as one of four breast cancer subtypes: Luminal A, Luminal B, Basal, or Her2-enriched.

We wish to obtain the entire hierarchy of joint structure among the four data blocks. Of particular interest are the four-way joint structure, three-way partially-shared joint structures, and partially-shared structure involving proteins or mutations. The biological understanding of the protein production pathway suggests that after accounting for the gene expression and gene copy number there should be no additional variation shared between mutations and proteins. Once all joint structure is cataloged, further conclusions can be drawn from loadings of the joint modes of variation. For example, the loadings from the gene expression data block would reveal which genes are involved in a certain cancer subtype if one of the joint modes of variation discriminates that subtype from the others. Note that DIVAS is an unsupervised method that does not make use of the class labels. Development of a supervised version of DIVAS remains an interesting open problem.  

Figure~\ref{fig:tcgascoresdiag} shows the DIVAS decomposition and angle diagnostics for the joint scores vectors for this data set. Detailed descriptions of the information plotted in the figure can be found in Section~\ref{sec:graphics}. 
DIVAS finds a single shared component between all four data blocks, a six-dimensional subspace shared between all data blocks besides mutation, and lots of shared structure between pairs of data blocks not involving mutation. This result aligns well with biological expectations, particularly the large amount of structure shared uniquely between gene expression and copy number. 


\begin{figure}[h]
\includegraphics[width=\textwidth]{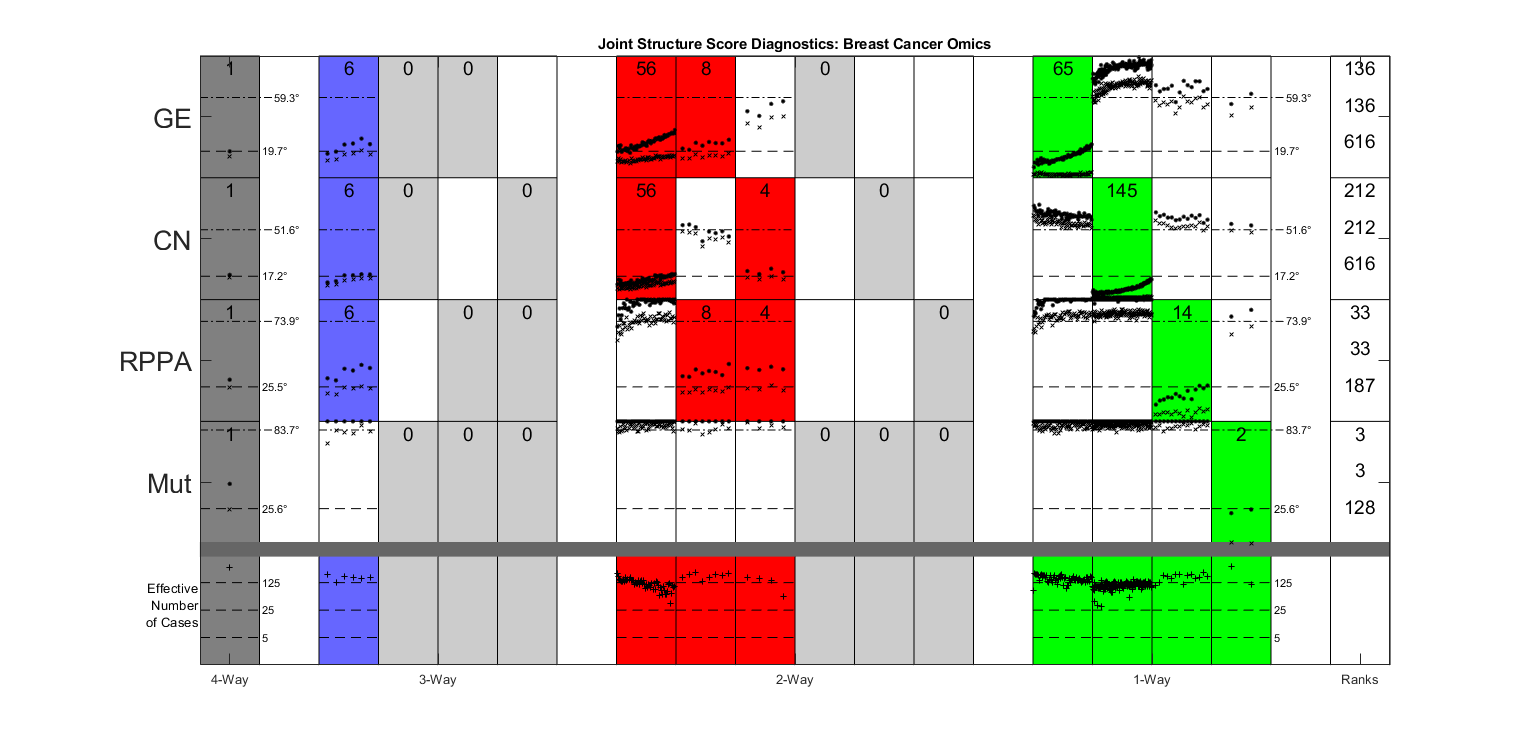}
\caption{Joint structure breakdown and diagnostics for the breast cancer omics score vectors.}\label{fig:tcgascoresdiag}
\end{figure}


Figure~\ref{fig:tcga4way3wayscores} displays a scores scatter plot matrix of directions from the four-way joint and three-way joint components of the data. In these plots, each point corresponds to a single data object. Each cancer subtype is shown with a different color and point symbol: basal with red triangles, luminal A with blue asterisks, luminal B with cyan x'es, and Her2-enriched with magenta pluses. On-diagonal plots show the coefficients of projection of data objects onto the score vector indicated. The vertical axis provides a jitter for ease of visual interpretation. Solid curves in on-diagonal plots are kernel density estimates, with the black curve including all data objects and the colored curves corresponding with each respective subtype. Off-diagonal plots show scatter plots of coefficients of projection of the data objects onto the score vectors in that plot's respective axes labels. More information about these plots may be found in Chapter~1 of \citet{marron2021object}. We chose to include the first two directions in the basis for the three-way joint subspace along with the four-way joint direction for ease of visual interpretation. The four-way joint component separates basal cases from other cases. This is typically the first component found in any analysis of the modes of variation in breast cancer patients, as basal cell cancers have a very different gene expression profile than the other subtypes. The two-dimensional plot of the scores along the first two directions in the three-way joint subspace separates Her2 cases from Basal cases primarily, and from Luminal A cases secondarily. This indicates the potential for identifying useful genes, proteins, and copy number regions that drive the variation in this three-way joint subspace that distinguish Her2 cases from other breast cancer subtypes.

\begin{figure}[h]
\includegraphics[width=\textwidth]{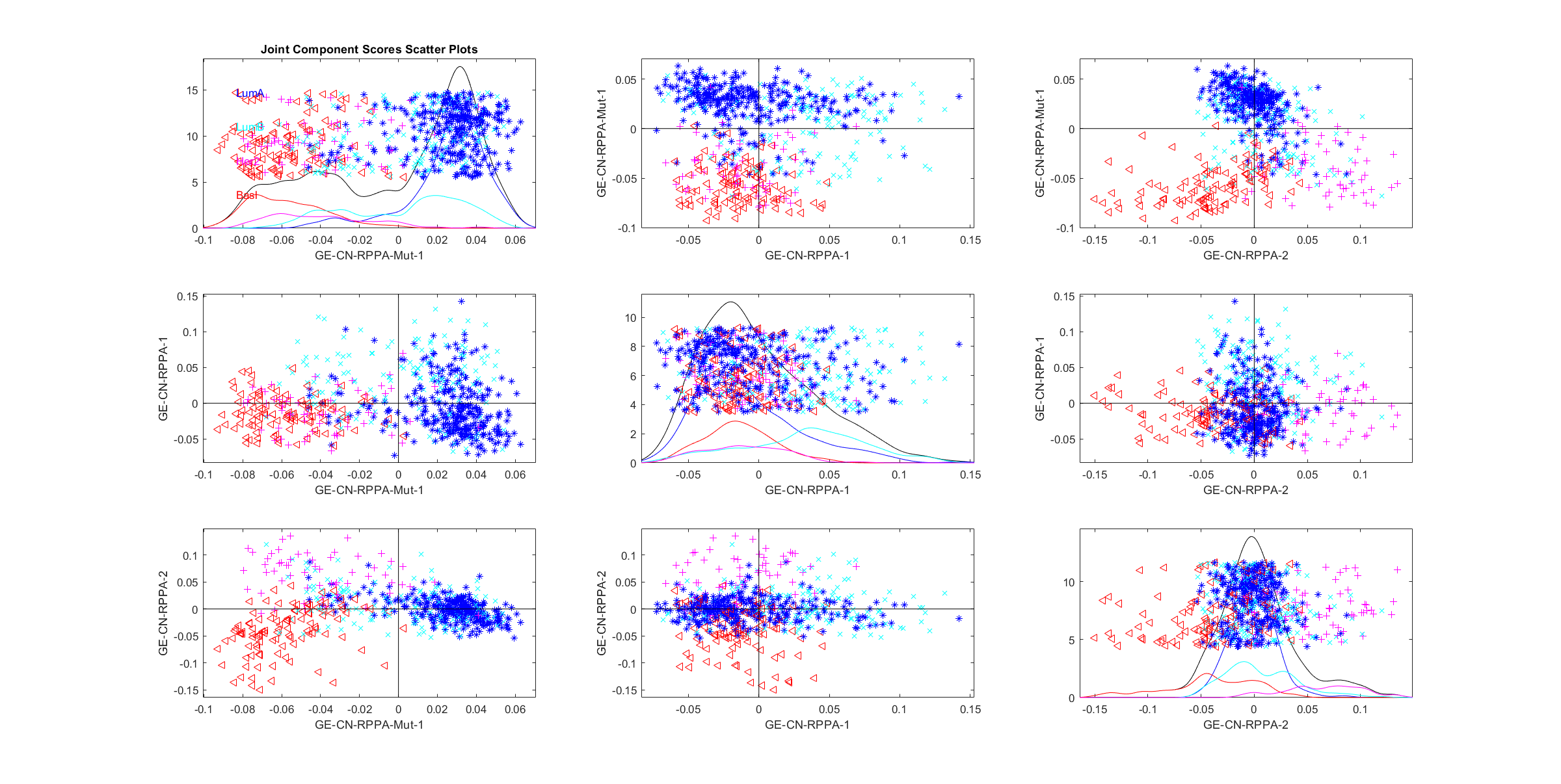}
\caption{Scores scatterplot matrix of four-way joint direction and first two three-way joint directions. The four-way joint subspace (top left) distinguishes basal (red triangles) from other subtypes and the three-way joint subspace (bottom right) distinguishes Her2 (magenta plusses) from other subtypes.}\label{fig:tcga4way3wayscores}
\end{figure}

\subsection{Case Study 2: Twentieth Century Mortality}\label{sec:data}








\citet{ooda} consider a data matrix containing mortality rates (proportion of the population of a given age that died in a year) of Spanish males from 1908 to 2002. We expand on this initial analysis by incorporating three additional data blocks: one for Spanish women and two more for Swiss men and women, and by increasing the end of the time frame to 2018. We are interested in how mortality rates changed over the course of this time period as a function of age. Hence, we will treat each year as a data object and the mortality rates of each age as a trait. We consider ages 12 to 90 to avoid zero counts for particularly high and low ages. Data was downloaded on April 8, 2021 from the Human Mortality Database  \citep{mortalitydatabase}.





To appropriately handle the multiple orders of magnitude present in mortality proportions, we transform each entry of the data blocks with a \textit{logit} function $f(x) = \log\left(\frac{x}{1-x}\right)$. After the logit transformation, each data block was \textit{double-centered}: the mean vectors in both object space and trait space were removed from each data block. As per the discussion in \citet{centering}, this type of centering is effective when all the data blocks share a common mode of variation in the trait mean direction. The object and trait means for each data block are displayed as curves in Figure~\ref{fig:mortswissmeans}. Each curve is colored according to year using a rainbow color scheme starting at magenta and blue, through orange and red. In the object mean panels (top row), we see the overall mortality profile across ages for each country and gender. Males exhibit a slightly higher increase in mortality upon entering adulthood than females in both countries due to increased risk-taking behaviors at that age \citep{mortsex2000, mortsex2009}. We also observe systematic anomalies in the mortality rates for older Spanish individuals that are not present in the Swiss data. As discussed in \citet{ooda}, these anomalies are manifestations of an age-rounding effect, and reflect major early differences in demographic record keeping practices between the two countries. The distinct rainbow sequence in each trait mean panel (bottom row) shows steady overall decreases in average mortality rate over time. The worst year of the 20th century flu pandemic, 1918, appears prominently at the top of each trait mean panel in violet. Spanish data block panels (especially males) have out-of-sequence light blue lines in their plots due to a civil war in the late 1930s.

\begin{figure}[h]
\includegraphics[width=\textwidth]{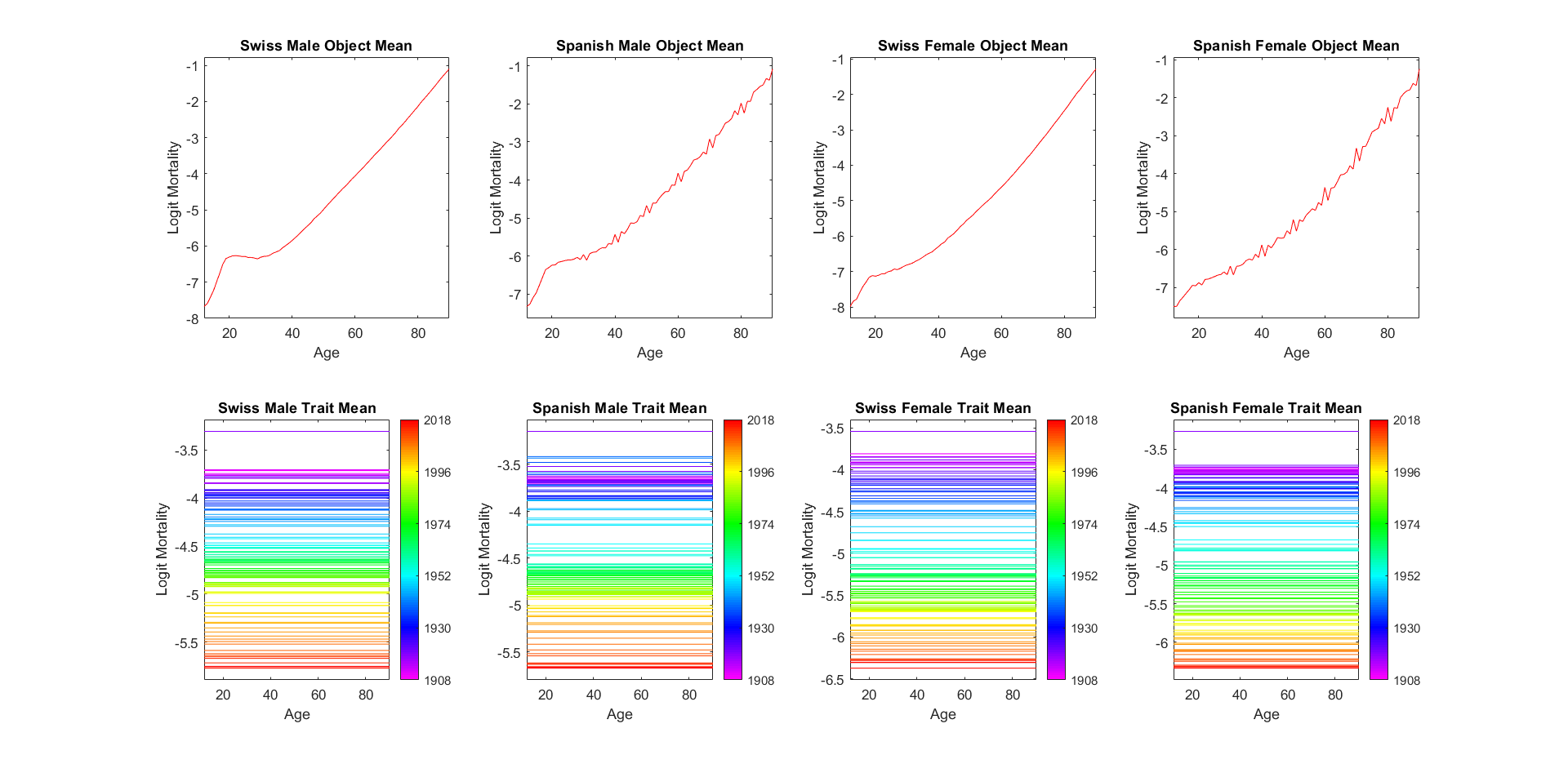}
\caption{Object means (top row) and trait means (bottom row) for each data block. Both male data blocks have a more dramatic increase in mortality for young adults than the female data blocks. Both Spanish data blocks display effects of record-keeping round-offs absent from Swiss data blocks. All trait means capture the overall improvement in mortality rate over time across the population.}\label{fig:mortswissmeans}
\end{figure}






Figure~\ref{fig:mortscoresdiag} shows the scores diagnostic graphic for the DIVAS decomposition of the mortality data, and Figure~\ref{fig:mortloadingsdiag} shows the angle diagnostics for the corresponding loadings vectors (see Section~\ref{sec:recon}). Since all four data blocks have identical trait and object dimensions and relatively similar variation, all the perturbation angles are also similar to each other. DIVAS finds a two-dimensional four-way shared component, a one-dimensional three-way component shared between each block besides Swiss females, and a six-dimensional shared component between Spanish males and females. Intuitive reasons for these findings are explained below via discussion of the modes of variation.


\begin{figure}[h]
\includegraphics[width=\textwidth]{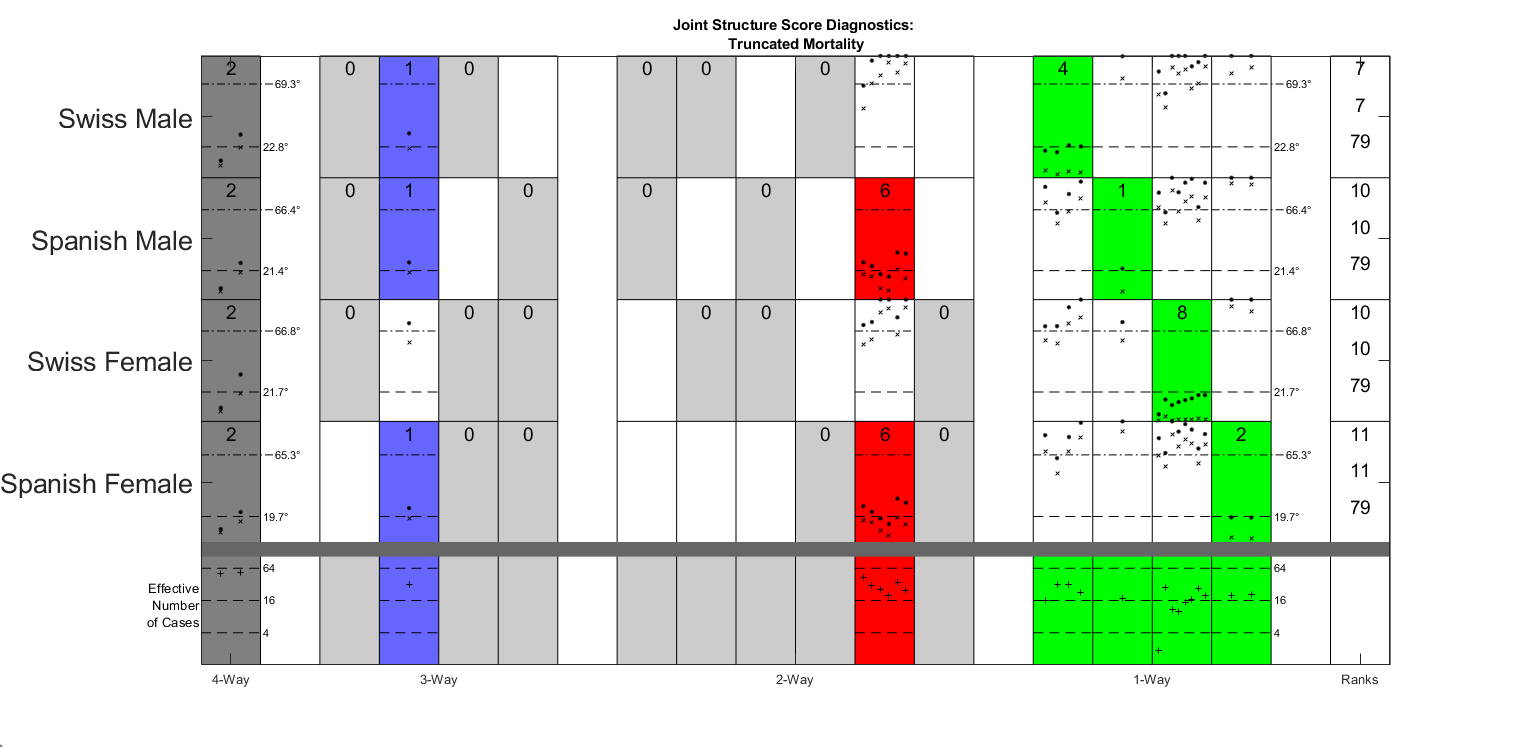}
\caption{Joint structure breakdown and diagnostics for the mortality data score vectors. Perhaps surprising is some amount of three-way partially shared joint structure. Spanish men and women have complex two-way partially shared joint structure due to age rounding.}\label{fig:mortscoresdiag}
\end{figure}

\begin{figure}[h]
\includegraphics[width=\textwidth]{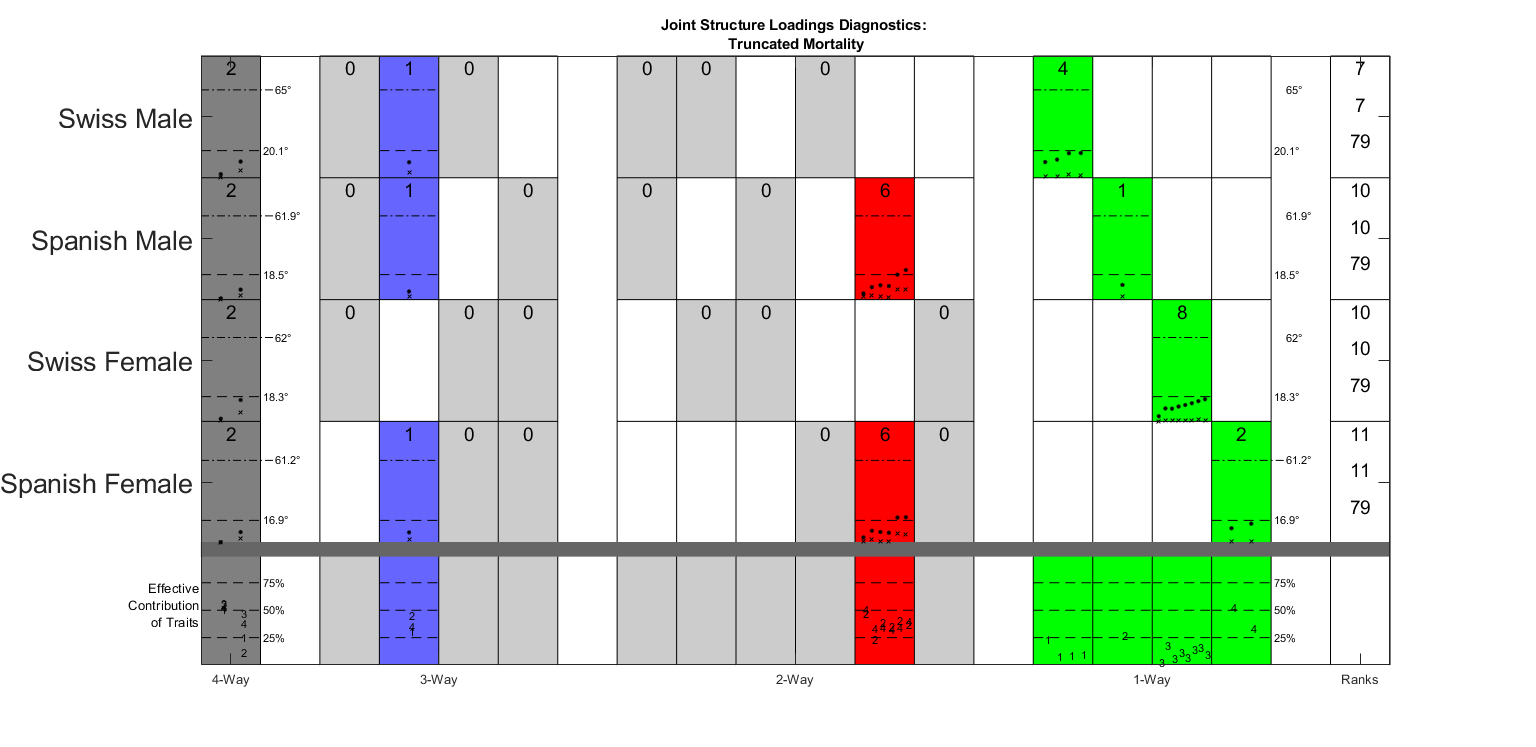}
\caption{Joint structure breakdown and diagnostics for the mortality data loadings vectors.}\label{fig:mortloadingsdiag}
\end{figure}





We further investigate the joint structure by visualizing the joint modes of variation about the mean in curve plots throughout Figures~\ref{fig:mortfourwayload}-\ref{fig:mortspanishtwowayload}. Figure~\ref{fig:mortfourwayload} shows such a visualization for the four-way joint structure. In this and subsequent mode of variation figures, each row of panels corresponds to a different basis direction and each column of panels corresponds to a data block. The final plot in each row shows the entries of the common score vector corresponding to that mode of variation. Each row of panels in this figure contains one trend in mortality that was found to be common across both countries and genders. The first mode is a contrast between older and younger individuals that manifests as a change in slope over time. In particular, while mortality rates decreased for all ages over the 20th century as per the trend seen in the trait mean, this mode of variation shows that decrease was more pronounced for younger individuals. The second mode is primarily a contrast between younger adults and middle-age adults that takes place between the 1970s and 1990s, with somewhat different age groupings across blocks. This contrast is the well-documented automotive safety effect described in \citet{ooda}. The wide proliferation of cars in the mid-20th century without modern safety guidelines in this time frame led to a notable increase in automobile fatalities concentrated in younger individuals. As automotive safety improved across Europe in the 1980s and 1990s, this source of excess mortality dissipated.


\begin{figure}[h]
\includegraphics[width=\textwidth]{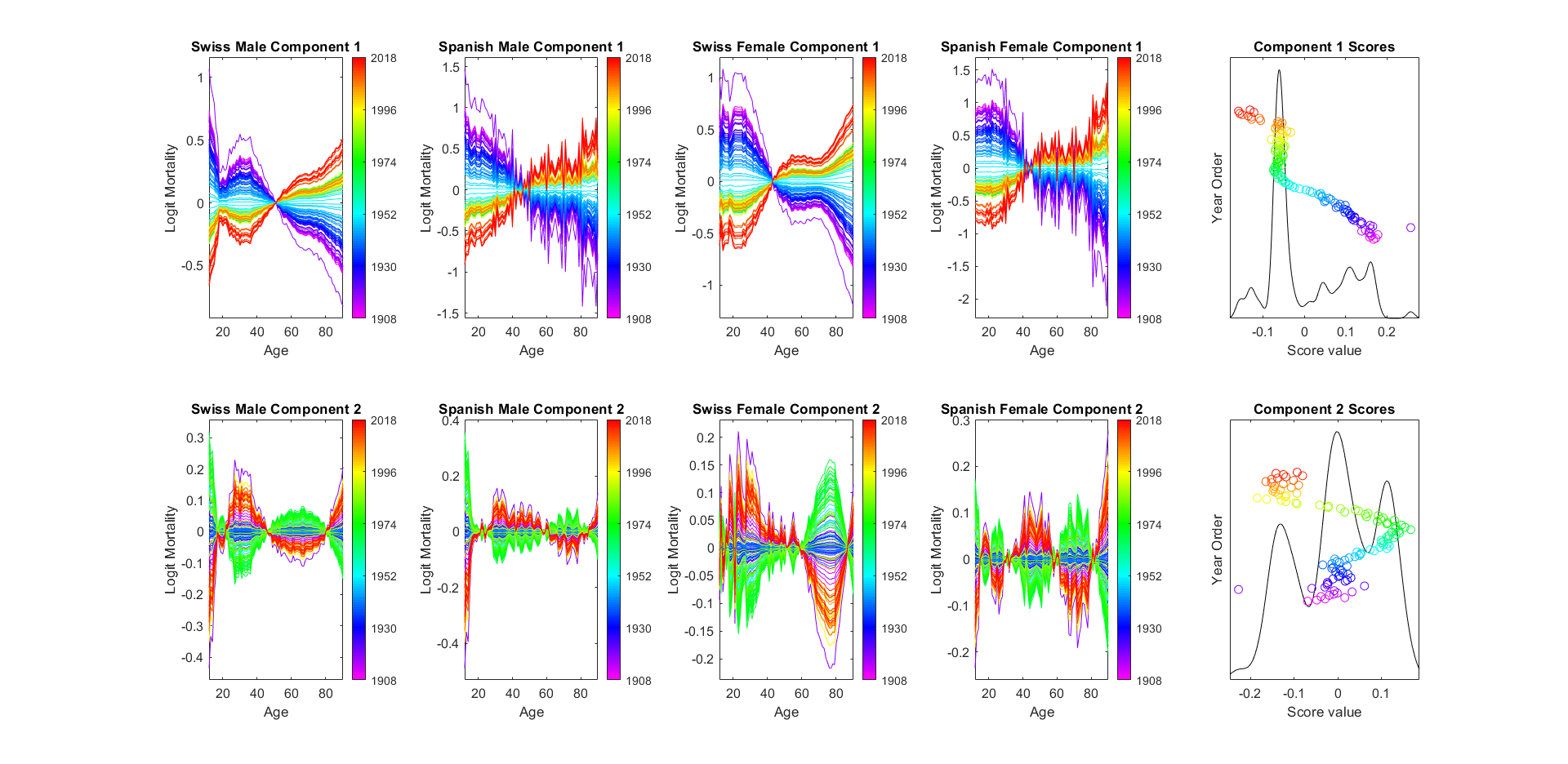}
\caption{Mode of variation curve plots of the two four-way joint components. Component 1 shows stronger improvement in mortality for younger people, and component 2 shows the rise and fall of automobile fatalities.} \label{fig:mortfourwayload}
\end{figure}

Figure \ref{fig:mortthreewayload} shows the single mode of variation found as three-way shared structure between Swiss men, Spanish men, and Spanish women. Before analyzing this data, we primarily expected to see four-way shared structure and pairwise shared structure across blocks with gender or country in common; this three-way shared structure deviates from that expectation. This mode of variation indicates a contrast between the mortality rates of young adults and the rest of the population during the late 1980s and early 1990s. Considering that the automotive safety effect already appeared in the fully joint component, we suspect this mode of variation is capturing a different phenomenon. Our hypothesis given the time frame and groups affected is that this component is capturing increased mortality from HIV/AIDS in the late 20th century. The contrast focusing on young males in both countries is the main driver of this hypothesis; the corresponding effect in Spanish women seems to be concentrated in older individuals so some additional effect might be entering the mode of variation in that block. This motivates further mortality research.

\begin{figure}[h]
\includegraphics[width=\textwidth]{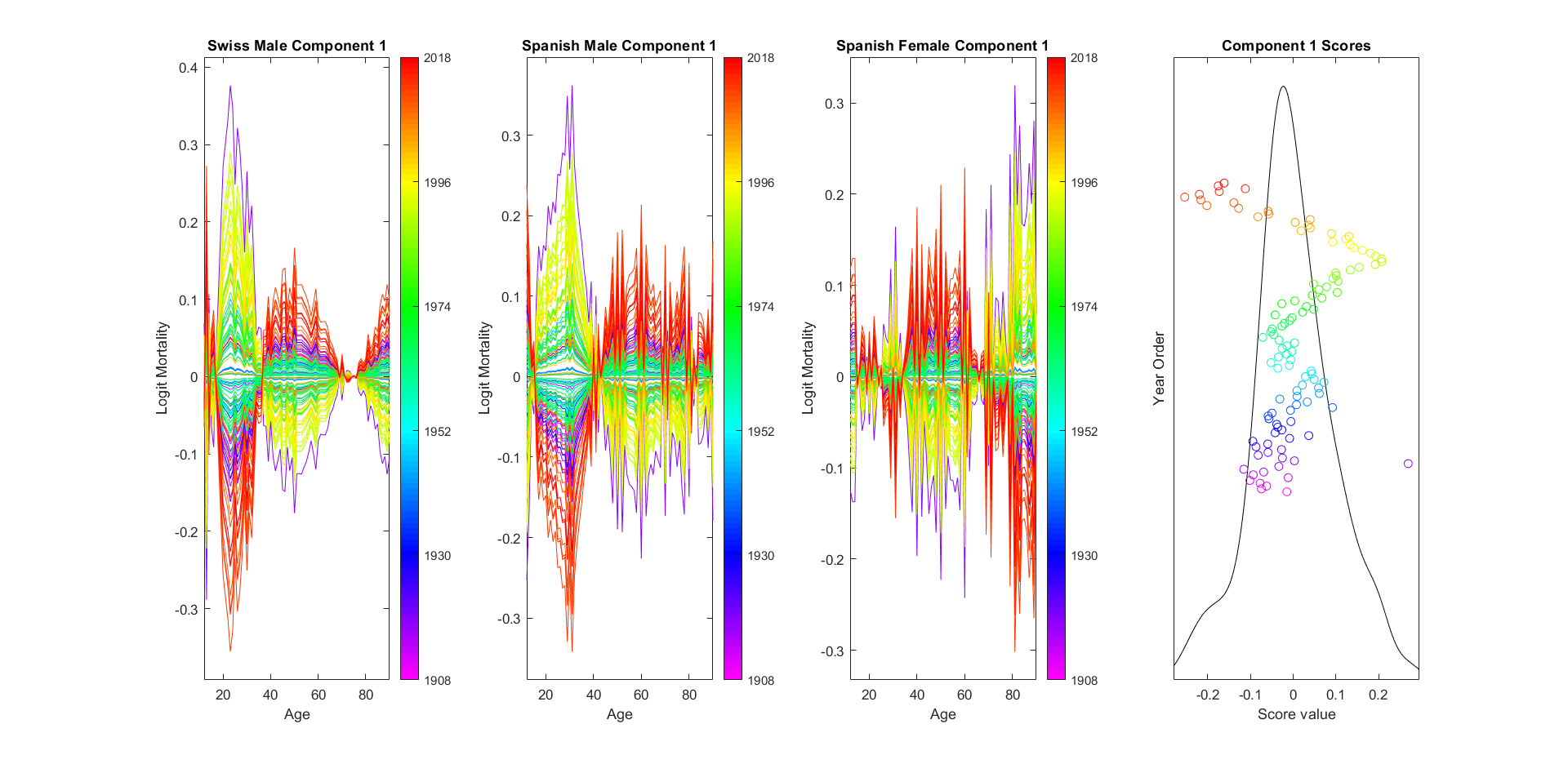}
\caption{Mode of variation curve plots of three-way joint component between Swiss men, Spanish men, and Spanish women. Mode of variation contrasts mortality in young adults with others around 1990. Potentially related to the emergence of HIV/AIDS.}\label{fig:mortthreewayload}
\end{figure}

Figure~\ref{fig:mortspanishtwowayload} shows the six modes of variation found as two-way joint structure shared between the two Spanish blocks. Component 1 captures excess mortality of young people, and especially young men, during the Spanish Civil War. Component 2 is a contrast within young adults that we do not fully understand. The remaining four modes of variation seem to be harmonic components generated by the age rounding effect discussed in \citet{ooda}.

\begin{figure}[h]
\includegraphics[width=\textwidth]{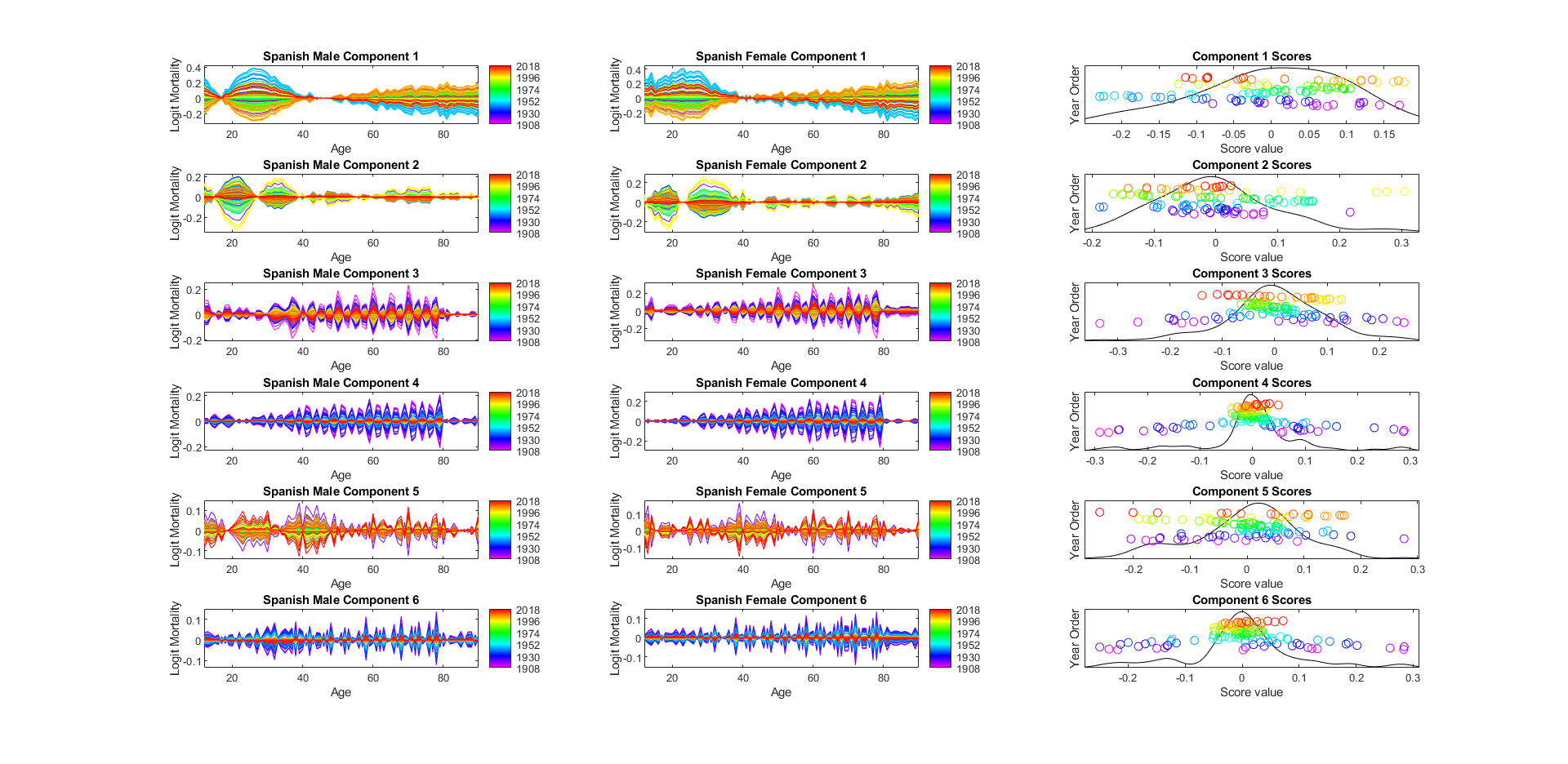}
\caption{Mode of variation curve plots of the two-way joint components between Spanish data blocks. First component driven by the Spanish Civil War in the 1930s. Last four components are primarily driven by high-frequency-harmonic age-rounding record-keeping anomalies.}\label{fig:mortspanishtwowayload}
\end{figure}

\section{Conclusions}\label{sec:conclusion}

This paper proposes DIVAS, a novel exploratory data analysis method for statistical data integration that allows for partially-shared structure between several distinct data blocks. The main contributions of DIVAS are twofold. First, we develop a rigorous, angle-based framework of statistical inference for diagnostically evaluating estimated shared structure. Second, we consider integration across both dimensions of the data blocks simultaneously and produce more thorough and higher-fidelity results as a consequence. 

Future work on DIVAS could proceed in at several different directions. Methodologically, there remains room for additional refinement of noise estimation throughout the first step of DIVAS, both in the noise variance estimator and the residual matrix estimator. Structurally, DIVAS is fundamentally a linear, unsupervised statistical model. Generalizations and expansions that extend DIVAS to tackle supervised learning and nonlinear relationships between data blocks are promising future directions. Practically, the driving force behind development of the method has always been appropriately complex data like the breast cancer omics data set that demands such methodological sophistication. Therefore we expect further improvements in DIVAS development will be found during analysis of ever more demanding data. 

\section{Acknowledgement}
Jan Hannig's research was supported in part by the National Science Foundation under Grant No. DMS-1916115, 2113404, and  2210337. J.S.~Marron's research was partially supported by the National Science Foundation Grant No. DMS-2113404.

\bibliographystyle{apalike}
\bibliography{dissRef}

\appendix
\section{Review of Random Matrix Theory}\label{sec:mp}


Our chosen signal extraction procedure uses random matrix theory ideas. The classical result from \citet{mp} on the distribution of the eigenvalues of random matrices underpins all of these ideas; we restate that result below.

Let $\Eb$ be a $d\times n$ random matrix. The entries of $\Eb$ are independent and identically distributed (i.i.d.) with mean 0, finite variance $\sigma^2$, and finite fourth moment. Form the $d\times d$ estimator of the covariance matrix $\Sigmab_n = \frac{1}{n}\Eb\Eb^{\top}$ and let $\lambda_1,\dots,\lambda_d$ denote the eigenvalues of $\Sigmab_n$. Consider the empirical measure $\mu_d(A) = \frac{1}{d}\#\{\lambda_j\in A\}, A\subset \bbR$ representing the empirical distribution of the eigenvalues of $\Sigmab_n$ as random variables themselves. Define an \textit{indicator fuction} $\1_{\{K\}}$ for a given condition $K$ as a function that returns 1 when condition $K$ is satisfied and returns 0 otherwise.

\begin{theorem}[\citealt{mp}]
If $d,n\rightarrow \infty$ such that $\frac{d}{n}\rightarrow \beta\in (0,+\infty)$, then $\mu_d$ converges weakly to the measure whose density is $\mu(\lambda)$:
\begin{equation}
\label{MPCases}
\mu(\lambda) = \begin{cases}
h(\lambda)\1_{(1-\sqrt{\beta})^2\leq \frac{\lambda}{\sigma^2}\leq (1+\sqrt{\beta})^2} & 0 < \beta \leq 1 \\
h(\lambda)\1_{(1-\sqrt{\beta})^2\leq \frac{\lambda}{\sigma^2} \leq (1+\sqrt{\beta})^2} + \left(1 - \frac{1}{\beta}\right)\1_{\lambda=0} & \beta > 1 
\end{cases}
\end{equation}

where the function $h(\lambda)$ is defined below:
\begin{equation}
\label{MPCurve}
h(\lambda) = \frac{1}{2\pi}\frac{\sqrt{\left((1+\sqrt{\beta})^2 - \frac{\lambda}{\sigma^2}\right)\left(\frac{\lambda}{\sigma^2} - (1-\sqrt{\beta})^2\right)}}{\beta \lambda}.
\end{equation}
\label{thm:mp}
\end{theorem}

If $d < n$, then $\beta < 1$ for $\Eb$ and $\Sigmab_n$ is rank $d$. In this case, since $\Sigmab_n$ is full-rank, all eigenvalues are nonzero, and asymptotically fall between $\sigma^2(1-\sqrt{\beta})^2$ and $\sigma^2(1+\sqrt{\beta})^2$. Alternatively, if $d > n$, then $\beta > 1$ for $\Eb$ and $\Sigmab_n$ is rank $n$. In this case $\Sigmab_n$ is not full rank so the eigenvalues $\lambda_{n+1} \dots \lambda_d$ are all 0. In cases where $\beta > 1$ the Marchenko-Pastur density is therefore a mixture between a point mass of $1 - \frac{1}{\beta}$ at zero and a continuous portion bounded between $\sigma^2(1-\sqrt{\beta})^2$ and $\sigma^2(1+\sqrt{\beta})^2$ with total area $\frac{1}{\beta}$.

\section{Review of Principal Angle Analysis}\label{sec:paa}

The following is based on \cite{paa} and \cite{MIAO199281}. \textit{Principal angle analysis} characterizes the relative positions of two subspaces $\mathcal{X}$ and $\mathcal{Y}$ in Euclidean space using canonical angles found via SVD. In particular, let $\Wb_{\mathcal{X}}$ and $\Wb_{\mathcal{Y}}$ be orthonormal basis matrices for $\mathcal{X}$ and $\mathcal{Y}$ respectively. Then the singular value decomposition of $\Wb_{\mathcal{X}}^{\top}\Wb_{\mathcal{Y}}$ finds both the principal angles between $\mathcal{X}$ and $\mathcal{Y}$ and the corresponding \textit{principal vectors}. Write the singular value decomposition of $\Wb_{\mathcal{X}}^{\top}\Wb_{\mathcal{Y}}$ as $\Wb_{\mathcal{X}}^{\top}\Wb_{\mathcal{Y}} = \Ub\Db\Vb^{\top}$, where $\Ub$ and $\Vb$ are orthonormal matrices containing the principal vectors of $\mathcal{X}$ and $\mathcal{Y}$ respectively, and $\Db$ is a diagonal matrix. The inverse cosines of the nonzero entries of $\Db$ give the principal angles between $\mathcal{X}$ and $\mathcal{Y}$, and in particular the angles between each pair of corresponding principal vectors. The $j$th pair of principal vectors have an angle between them equal to the $j$th principal angle. 


This perspective also demonstrates the result of principal angle analysis when the dimensions of $\mathcal{X}$ and $\mathcal{Y}$ differ. Let the dimensions of $\mathcal{X}$ and $\mathcal{Y}$ be $p$ and $q$ respectively, with $p < q$. In this case some of the singular values will be zero as the matrix $\Wb_{\mathcal{X}}^{\top}\Wb_{\mathcal{Y}}$ is non-square, and the inverse cosine of zero is $90^{\circ}$. If $p<q$, the principal angles $\theta_{p+1},\dots,\theta_{q}$ are all $90^{\circ}$.

Principal angle analysis is also orthogonally invariant. In particular, the principal angles between $\mathcal{X}$ and $\mathcal{Y}$ will be identical to the principal angles between reoriented versions $\Ob\mathcal{X}$ and $\Ob\mathcal{Y}$, where $\Ob$ is an orthogonal matrix and $\Ob\mathcal{X} = \{\Ob\xb|\xb\in\mathcal{X}\}$. The matrices $\Ob \Wb_{\mathcal{X}}$ and $\Ob \Wb_{\mathcal{Y}}$ represent orthonormal bases for $\Ob\mathcal{X}$ and $\Ob\mathcal{Y}$, so the principal angle structure between the two rotated subspaces is found by taking a singular value decomposition of $\Wb_{\mathcal{X}}^{\top} \Ob^{\top} \Ob \Wb_{\mathcal{Y}}$, which is equivalent to that of $\Wb_{\mathcal{X}}^{\top}\Wb_{\mathcal{Y}}$.

\section{Noise Matrix Estimation} \label{sec:noise}

The residual $\hat{\Eb} = \Xb - \hat{\Ab}$ is a poor estimate of the non-signal component of the data, especially in the case of a non-square matrix. Heuristically, this is caused by the residual lying entirely in the subspace spanned by the data. We investigate the causes of this phenomenon and propose a solution.

For this investigation our synthetic data will be a $5000\times 500$ matrix $\Xb=\Ab+\Eb$. The signal matrix $\Ab=\Ub\Db\Vb^{\top}$ is rank 50 with equally-spaced singular values from 0.1 to 5. $\Eb$ is a full-rank i.i.d. Gaussian matrix with variance $\frac{\sigma^2}{5000}$. This scaling of the noise variance by the number of traits is common in the matrix signal processing literature \citep{gd2014, gd2017}. It results in columns with expected norm $\sigma$ and sets the noise at a level commensurate with the magnitude of the signal. With $\sigma=1$, we expect most of the singular values to be easily recoverable while others are indistinguishable from the noise. We perform signal extraction as described in Section~\ref{sec:mpsignal} on this matrix and subsequently examine estimators of $\Eb$ given $\hat{\Ab}$.

One way to check the efficacy of an estimator $\hat{\Eb}$ for $\Eb$ is to see how well its eigenvalues align with the Marchenko-Pastur distribution (see Appendix~\ref{sec:mp}). We can compare the observed values to theoretical quantiles using a \textit{quantile-quantile} (Q-Q) plot. On the horizontal axis we plot the sorted observed eigenvalues for a noise matrix estimate $\hat{\Eb}$, and on the vertical axis we plot evenly-spaced quantiles of the Marchenko-Pastur distribution with parameter $\beta=\frac{d\wedge n}{d\vee n}$ (here $d$ and $n$ are the row and column dimensions of the matrix respectively). Typically, if the plotted points on a Q-Q plot roughly follow the $45^{\circ}$ line, the conclusion is that the observed data aligns well with the theoretical distribution. To get a sense of how much variability to expect about the $45^{\circ}$ line, we generate $M=100$ i.i.d. Gaussian matrices and plot their eigenvalues as green lines underneath the magenta Q-Q points. These traces create a visually striking region of acceptable variability which can be used to judge the goodness of fit at a glance.

Figure~\ref{nonsquareEigQQ} shows such a Q-Q plot for the eigenvalues of the na\"ive noise estimate $\hat{\Eb}=\Xb - \hat{\Ab}$ for the synthetic data matrix. The na\"ive estimated non-signal component tends to display perhaps unexpectedly low energy in directions associated with the estimated signal subspace. This phenomenon leads to Q-Q plots that are challenging to interpret. For this matrix the estimated signal rank is 44, and the bottom 44 eigenvalues of the estimated noise matrix completely deviate from the theoretical Marchenko-Pastur distribution. Importantly, this phenomenon (explained in detail below) occurs regardless of the chosen estimate for $\sigma$. The aberration in this graphic demonstrates the ineffectiveness of the na\"ive estimate. The rotational bootstrap procedure (see Section~\ref{sec:rotboot}) central to DIVAS depends on effective estimation of the underlying noise matrix $\Eb$. This is accomplished via a correction to a portion of the singular values of $\hat{\Eb}$.

\begin{figure}[h]
\centering
\includegraphics[scale=0.4]{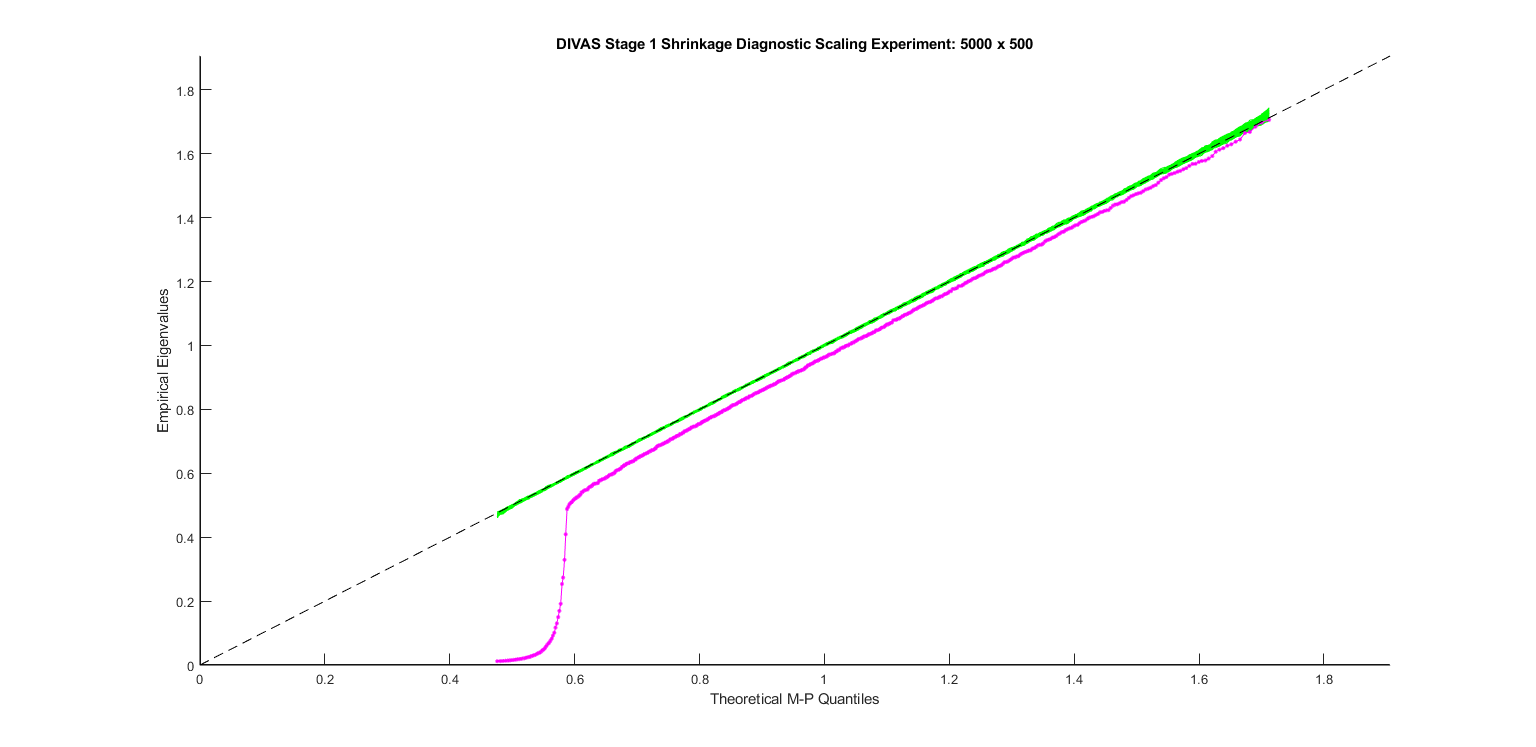}
\caption{Q-Q plot for the eigenvalues of the na\"ive noise matrix estimate. The first $\hat{r}$ eigenvalues fall entirely outside the range determined by Theorem \ref{thm:mp}, signaling that the naive noise matrix estimate is flawed. These eigenvalues are also scaled using the original noise level estimate to retain some interpretability. Scaling using the apparent noise level in the estimated error matrix produces an even worse fit to the Marchenko-Pastur distribution because the apparent noise level is too low.}\label{nonsquareEigQQ}
\end{figure}

To explain this behavior and motivate our proposed correction, we consider our data model \eqref{eq:datamodel} in a special case where the signal is rank one, and the signal, noise and data are all vectors in $\bbR^2$, illustrated in Figure~\ref{angles}. The signal (green) and noise (red) vectors each lie in distinct one-dimensional subspaces. The two vectors are added together to form the data vector (blue). When we form our estimate of the signal $\hat{\Ab}$ (green-blue dashed) our shrinkage procedure gives us a good estimate of signal magnitude. However, it is challenging to recover directional information about $\Ab$ as the estimate $\hat{\Ab}$ lies in the same subspace as the data. When we next subtract $\hat{\Ab}$ from the data $\Xb$ to form the na\"ive estimate $\hat{\Eb}$ (red-blue dashed), the subtraction occurs entirely in the data subspace so we don't account for the angle between the initial signal and noise vectors at all. This leads to an underestimation of noise energy: the length of the estimated noise within the data subspace (red-blue dashed) is distinctly shorter than the length of the original noise vector (red). This length discrepancy is the one-dimensional analog of the phenomenon shown in Figure~\ref{nonsquareEigQQ} where many of the smallest eigenvalues are even smaller than expected.

\begin{figure}[h]
\begin{minipage}[c]{0.49\linewidth}
            \centering
            \includegraphics[width=\textwidth]{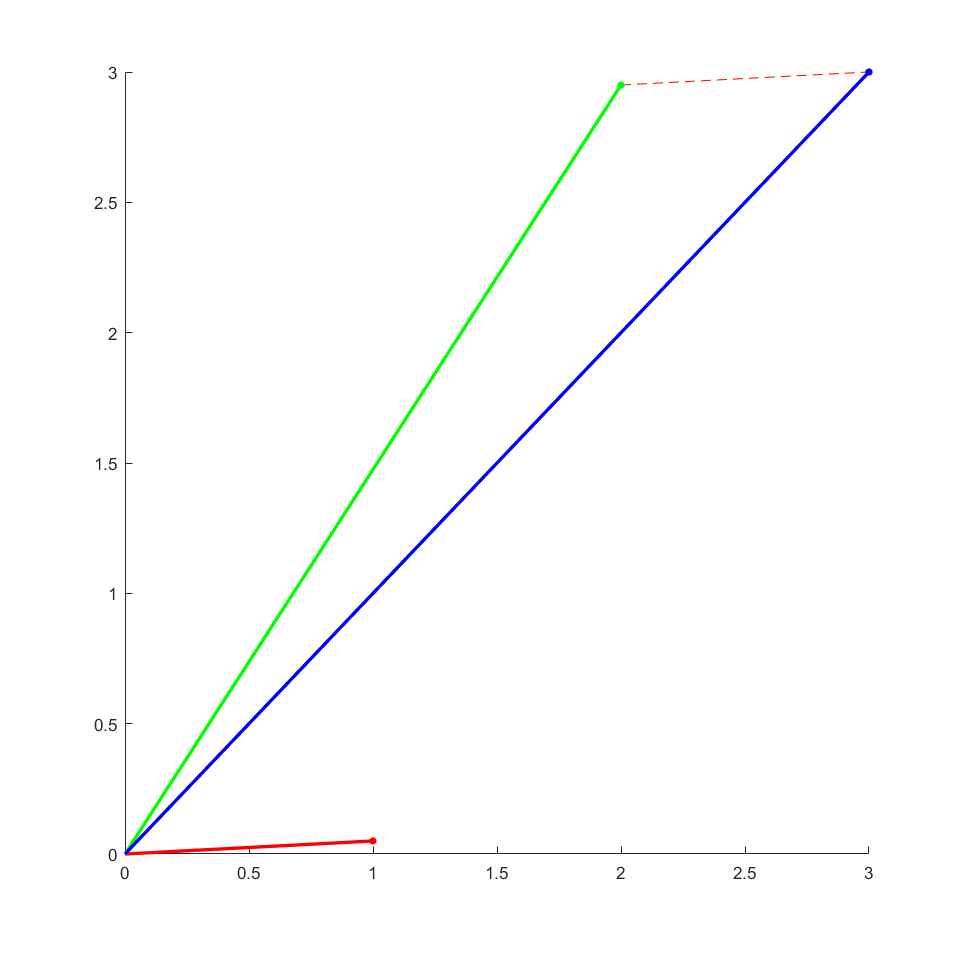}
        \end{minipage}
        \begin{minipage}[c]{0.49\linewidth}
            \centering
            \includegraphics[width=\textwidth]{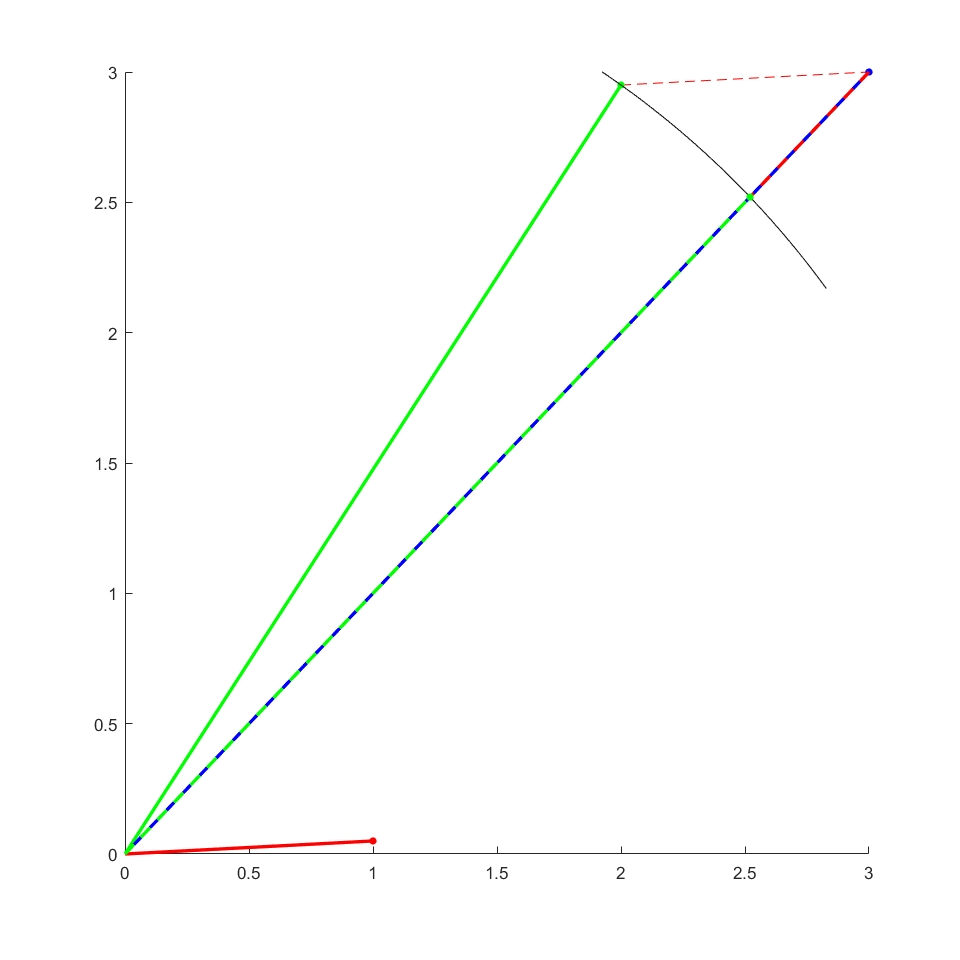}
        \end{minipage}
\caption{Example of noise energy underestimation for a rank-one signal subspace in $\bbR^2$. Left: Signal space (green), noise space (red), and data space (blue). Data vector formed by adding signal and noise vectors tip to tail. Right: Estimating $\hat{\Ab}$ (green-blue dashed) and $\hat{\Eb}=\Xb-\hat{\Ab}$ (red-blue dashed). When we remove energy equal to that of the signal space from the data space the leftover energy is noticeably smaller than the true noise energy. Note that the black arc indicates a rotation rather than a projection, so the green-blue dashed line has the same length as the green line.} \label{angles}
\end{figure}

Several potential corrections for this effect are proposed in Chapter~4 of \citet{jbpdiss}. We present the correction used in DIVAS here. Consider $\Xb$ as a sum of rank 1 approximations in the manner of \eqref{eq:signalest}: $\Xb= \sum_{i=1}^{d\wedge n} \overbar{\nu}_i \overbar{\ub}_i \overbar{\vb}_i^{\top}$. Once we estimate the signal singular values, we can split the energy in the associated singular vector directions into signal energy $\hat{\nu}$ and non-signal energy $\overbar{\nu}-\hat{\nu}$:
\begin{equation}
    \Xb= \sum_{i=1}^{\hat{r}} \hat{\nu}_i \overbar{\ub}_i \overbar{\vb}_i^{\top} +  \sum_{i=1}^{\hat{r}} (\overbar{\nu}_i - \hat{\nu}_i) \overbar{\ub}_i \overbar{\vb}_i^{\top} + \sum_{i=\hat{r}+1}^{d\wedge n} \overbar{\nu}_i \overbar{\ub}_i \overbar{\vb}_i^{\top}
\end{equation}
The Gavish-Donoho shrinkage function \eqref{eq:opShrink} gives us good estimates for the first $\hat{r}$ singular values $\hat{\nu}_{1:\hat{r}}$ while confirming many of the $\overbar{\ub}_i \overbar{\vb}_i^{\top}$ subspaces and associated singular values as noise. However, by subtracting $\hat{\Ab} = \sum_{i=1}^{\hat{r}} \hat{\nu}_i \overbar{\ub}_i \overbar{\vb}_i^{\top}$ from $\Xb$ we are overestimating the influence of the signal within the data subspace as the $\overbar{\nu}_i - \hat{\nu}_i$ terms of $\hat{\Eb}$ have inordinately low energy in directions associated with the estimated signal.

The DIVAS solution to this energy deficiency is to replace each deficient singular value in $\hat{\Eb}$ with a Marchenko-Pastur random variate. Let $MP_q(\beta)$ be the $q$th percentile of the Marchenko-Pastur distribution with parameter $\beta$, let $U_{1:\hat{r}}$ be $\hat{r}$ i.i.d. standard uniform random variables, and let $\hat{\sigma}^2$ be an estimate of the noise variance. We form the \textit{imputed} noise matrix estimate $\hat{\Eb}_{impute}$ as follows:
\begin{equation}\label{eq:ehatimpute}
    \hat{\Eb}_{impute} =
    \sum_{i=1}^{\hat{r}} \hat{\sigma} MP_{U_i}(\beta) \overbar{\ub}_i \overbar{\vb}_i^{\top} + \sum_{i=\hat{r}+1}^{d\wedge n} \overbar{\nu}_i \overbar{\ub}_i \overbar{\vb}_i^{\top}
\end{equation}

Figure~\ref{nonsquareQQImproved} shows the original Q-Q plot from Figure~\ref{nonsquareEigQQ} with the eigenvalues of $\hat{\Eb}_{impute}$ for the synthetic data matrix also included in black. After imputing the deficient singular values, the eigenvalues of the reconstructed noise matrix estimate follow the expected Marchenko-Pastur distribution quite closely; nearly all of them fall within the green acceptable variability envelope.

\begin{figure}[h]
\centering
\includegraphics[scale=0.4]{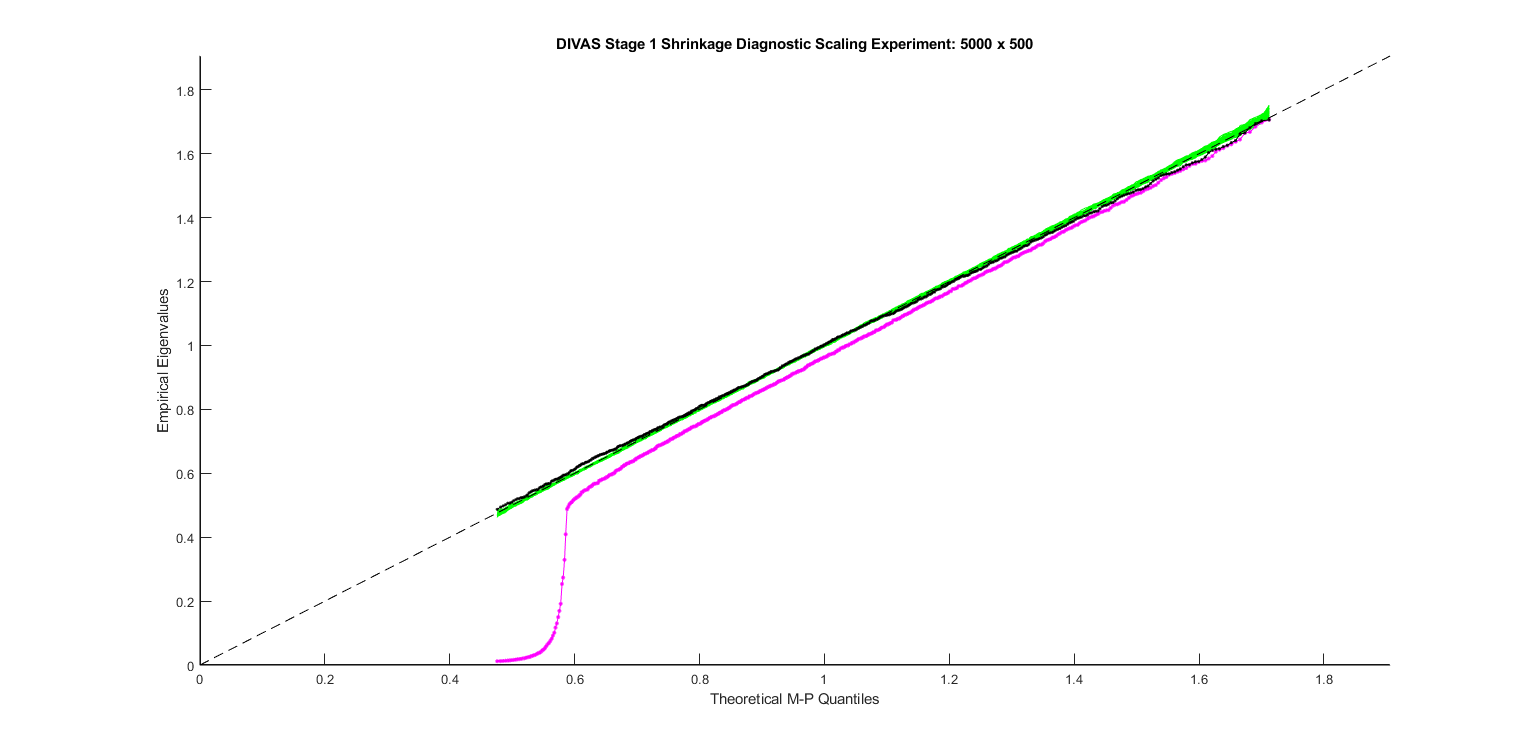}
\caption{Q-Q plot for the eigenvalues of the na\"ive noise matrix estimate (magenta) and the imputed noise matrix estimate (black) from the synthetic data matrix. The corrected eigenvalues largely remain within the green acceptable variability envelope.}\label{nonsquareQQImproved}
\end{figure}

\section{Optimization Algorithm and Implementation}\label{sec:optdetails}
In this appendix we provide the details of our numerical algorithm to solve the optimization problem \eqref{DivasOptHeur}. 
First, we will explicitly rewrite \eqref{DivasOptHeur} into a convex-concave optimization problem, also called a DC (difference of two convex functions) program.
The detail primarily involves the steps to reformulate the problem \eqref{DivasOptHeur} in terms of the convex-concave setting described in \citet{il2016}, and subsequently implements that convex-concave procedure for solving the resulting problem. 
The convex-concave procedure (or also called a DC algorithm) can be found in the literature including \citet{il2016,Tran-Dinh2009b}.
Since our problem has both the DC objective function and DC constraints, we can use the convergence analysis  from \cite{Tran-Dinh2009b} to guarantee the well-definedness of our algorithm.

\vspace{1ex}
\noindent\textbf{DC programming reformulation of \eqref{DivasOptHeur}.}
To move towards to a DC programming reformulation of \eqref{DivasOptHeur}, we express the angles between candidate directions $\vb\astar$ and various subspaces in terms of their squared cosines. 
For an arbitrary-magnitude $\vb\astar$ and orthonormal basis matrix $\Vb$ for a subspace, if we define $\hat{\theta}_{\Vb}=\angle(\vb\astar, \Vb)$, then we have $\cos^2(\hat{\theta}_{\Vb}) = \frac{\vb\astart\Vb\Vb^{\top}\vb\astar}{\vb\astart\vb\astar}$.
More specifically, by using the representation  $\cos^2(\hat{\theta}_{Tk}) = \frac{\vb\astart \check{\Vb}_k \check{\Vb}_k^{\top}\vb\astar}{\vb\astart\vb\astar}$ and keeping in mind the orthonormal condition that $\vb\astart\vb\astar = 1$, the objective function of  \eqref{DivasOptHeur} becomes $-\sum_{k\in\ib} \cos^2(\hat{\theta}_{Tk}) = -  \vb\astart\left(\sum_{k\in \ib} \check{\Vb}_k \check{\Vb}_k^{\top}\right)\vb\astar$.
Next, using the decreasing monotonicity of $\cos^2$ in $[0, \frac{\pi}{2}]$, the constraint $\hat{\theta}_{Tk} \leq \hat{\phi}_k$ is equivalent to $\cos^2(\hat{\theta}_{Tk}) = \frac{\vb\astart\check{\Vb}_k \check{\Vb}_k^{\top}\vb\astar}{  \vb\astart\vb\astar }  \geq \cos^2(\hat{\phi}_k)$ for all $k\in\ib$.
Similarly, $\hat{\theta}_{Tk} \geq \hat{\phi}_k$ is equivalent to $\cos^2(\hat{\theta}_{Tk}) = \frac{\vb\astart\check{\Vb}_k \check{\Vb}_k^{\top}\vb\astar}{ \vb\astart\vb\astar } \leq \cos^2(\hat{\phi}_k)$ for all $k\in\ib^c$.
The constraint $\hat{\theta}_{Ok} \leq \hat{\psi}_k$ is equivalent to $\cos^2(\hat{\theta}_{Ok}) = \frac{  \vb\astart\Xb_k^{\top}\check{\Ub}_k \check{\Ub}_k^{\top} \Xb_k \vb\astar }{ \vb\astart\Xb_k^{\top}\Xb_k\vb\astar }  \geq \cos^2(\hat{\psi}_k)$ for all $k\in\ib$.
Finally, we multiply all these constraint reformulations by $\vb\astart\vb\astar$ to eliminate their denominator, and transform them into DC constraints.
The orthonormal constraint $\vb\astart\vb\astar = 1$ is equivalent to $\vb\astart\vb\astar - 1 \leq 0$ and $1 - \vb\astart\vb\astar \leq 0$.
Putting these transformations together, we can easily see that \eqref{DivasOptHeur} is equivalent to the following  problem:
\begin{equation}\label{DivasOptHeur_equiv_form}
\left\{\begin{array}{llllll} 
&\min_{\vb\astar}& - \vb\astart\left(\sum_{k\in \ib} \check{\Vb}_k \check{\Vb}_k^{\top}\right)\vb\astar  & & & \\
&s.t.&  \hat{\theta}_{Tk} &= \angle(\vb\astar,\check{\Vb}_k) & \forall k & \\
&  & \hat{\theta}_{Ok} &= \angle(\Xb_k\vb\astar,\check{\Ub}_k) & \forall k & \\
&  &\cos^2(\hat{\phi}_k)\vb\astart\vb\astar - \vb\astart\check{\Vb}_k \check{\Vb}_k^{\top}\vb\astar &\leq 0 & \forall k &\in \ib \\
&  &\vb\astart\check{\Vb}_k \check{\Vb}_k^{\top}\vb\astar - \cos^2(\hat{\phi}_k)\vb\astart\vb\astar&\leq 0 & \forall k &\in \ib^c \\
&  &\cos^2(\hat{\psi}_k)\vb\astart\Xb_k^{\top}\Xb_k\vb\astar - \vb\astart\Xb_k^{\top}\check{\Ub}_k \check{\Ub}_k^{\top} \Xb_k \vb\astar &\leq 0 & \forall k &\in \ib \\
&  &\vb\astar &\perp \mathfrak{V}_{\jb} & \forall \jb &\supseteq \ib \\
& &\vb\astart\vb\astar -1&\leq 0 & & \\
& &1 - \vb\astart\vb\astar &\leq 0 & &
\end{array}\right.
\end{equation}
Clearly, the objective function of \eqref{DivasOptHeur_equiv_form} is concave.
In addition, the third, the fourth, the fifth, and the last constraints of \eqref{DivasOptHeur_equiv_form} are DC constraints of the form $f(\vb\astar) - g(\vb\astar) \leq 0$.
Therefore, \eqref{DivasOptHeur_equiv_form} is a DC program.
This problem is feasible depending on the choice of $\hat{\phi}_k$ and $\hat{\psi}_k$.
However, due to the orthonormal constraint $\vb\astart\vb\astar = 1$, the feasible set of problem \eqref{DivasOptHeur_equiv_form} does not have nonempty interior. 
In this case, to guarantee the DC algorithm being well-defined, we will relax it by adding slack variables.

\vspace{1ex}
\noindent\textbf{DC algorithm.}
Note that the objective function of \eqref{DivasOptHeur_equiv_form}  is though concave, it can be written into a DC function $f_0(\vb\astar) - g_0(\vb\astar)$, where $f_0 = 0$ and $g_0$ is a quadratic function.
Assume that we have $m_c$ DC constraints.
Then, all the DC constraints can be written as $f_k(\vb\astar) - g_k(\vb\astar) \leq 0$ for $k=1,\cdots, m_c$.
The other convex constraints are expressed as $\vb\astar \in \mathcal{F}$, including the orthogonal constraints $\vb\astar \perp \mathfrak{V}_{\jb}$ for all $\jb \supseteq \ib$, which are in fact linear.
However, to guarantee the feasibility of our DC program, we instead relax the DC constraints to obtain $f_k(\vb\astar) - g_k(\vb\astar) \leq s_k$, where $s_k \geq 0$ are given slack variables.
We also penalize the slack variables $s_k$ into the objective function with a given penalty parameter $\tau > 0$ to better approximate feasible solutions of \eqref{DivasOptHeur_equiv_form}.
Therefore, we can write the relaxation form of \eqref{DivasOptHeur_equiv_form} into the following DC program:
\begin{equation}\label{eq:DC_program}
\left\{\begin{array}{lll}
{\displaystyle\min_{\vb\astar, s_k}} & f_0(\vb\astar) - g_0(\vb\astar) + \tau \sum_{k=1}^{m_c}s_k\vspace{1ex}\\
\mathrm{s.t.} & f_k(\vb\astar) - g_k(\vb\astar) \leq s_k, \quad \forall i =1,\cdots, m_c,\\
&\vb\astar \in \mathcal{F}, \quad s_k \geq 0, ~(k=1,\cdots,m_c).
\end{array}\right.
\end{equation}
Note that if $s_k = 0$ for $k=1,\cdots, m_c$, then \eqref{eq:DC_program} reduces to \eqref{DivasOptHeur_equiv_form}.
To solve \eqref{eq:DC_program}, we apply a DC algorithm (see, e.g., \citet{il2016,Tran-Dinh2009b}), which can be roughly described as follows.
\begin{enumerate}
    \item \textit{Initialization:} At the iteration $t=0$, find an initial point $\vb_0$ of \eqref{eq:DC_program} (specified later).
    \item \textit{Iteration $t$.} At each iteration $t\geq 0$, given $\vb_t$, linearize the concave parts of \eqref{eq:DC_program} to obtain the following convex optimization subproblem:
\begin{equation}\label{eq:cvx_subprob}
\left\{\begin{array}{lll}
{\displaystyle\min_{\vb\astar, s_k}} & f_0(\vb\astar) - [g_0(\vb_t) + \nabla{g_0}(\vb_t)^{\top}(\vb - \vb_t)] + \tau \sum_{k=1}^{m_c}s_k\vspace{1ex}\\
\mathrm{s.t.} & f_k(\vb\astar) - [g_k(\vb_t) + \nabla{g_k}(\vb_t)^{\top}(\vb - \vb_t)] \leq s_k, \quad (k =1,\cdots, m_c),\\
&\vb\astar \in \mathcal{F}, \quad s_k \geq 0, ~(k=1,\cdots,m_c).
\end{array}\right.
\end{equation}
\item[] Solve \eqref{eq:cvx_subprob} to obtain an optimal solution $\vb_{t+1}$ and repeat the next iteration $t+1$ with $\vb_{t+1}$.
\item \textit{Termination.} The algorithm is terminated if it does not significantly improve the objective values, or other criteria are met.
\end{enumerate}
Note that as proven in  \cite{Tran-Dinh2009b}, under mild conditions imposed on \eqref{eq:DC_program}, our DC algorithm guarantees that the sequence $\{\vb_t\}$ generated by our DC algorithm converges to a stationary point of \eqref{eq:DC_program} (i.e. the point satisfying the optimality condition of \eqref{eq:DC_program}).
We do not repeat the convergence analysis of our DC procedure here, but refer to  \cite{Tran-Dinh2009b} for more details.

\vspace{1ex}
\noindent\textbf{Detailed implementation.}
We now specify the detailed implementation of our DC procedure as follows.
The first step is to choose an initial point $\vb_0$ for our DC program \eqref{DivasOptHeur_equiv_form}.
Without relaxation, choosing a feasible initial point for \eqref{DivasOptHeur_equiv_form} is indeed challenging. 
Hence, we introduce slack variables $s_k$ to the constraints to guarantee that our relaxed DC program is always feasible and thus our algorithm is well-defined and can proceed. 
For instance, one can directly choose an arbitrary $\vb_0$ in $\mathcal{F}$ first, and then set $s_{k, 0} = \max\{ f_k(\vb_0) - g_k(\vb_0), 0\}$ for each $k$ to obtain a feasible point of \eqref{DivasOptHeur_equiv_form}.

In our implementation, we choose as our initial point for the $i$th direction in the joint subspace of block collection $\ib$ the $i$th right singular vector from the SVD of $\left[\Vb_{k}\right]^{\top}_{k\in\ib}$, which is related to the joint structure found via the AJIVE algorithm \citep{ajive}. 
If necessary, the chosen initial point is also projected to obey any orthogonality constraints present at that point in the algorithm.
As explained, the slack variables $s_k$ introduced to allow for an infeasible initial condition also appear in the objective function. 
They are penalized with a weight $\tau$ (also called the penalty parameter) which changes on each iteration of the optimization problem as $\tau_t$. 
Notably, the values of the quadratic forms involved in the object space constraints are often much larger than those for the trait space constraints as the object space constraints include the full-energy data matrices $\Xb_k$. 
Therefore, we downweight the slack penalty on those constraints by the leading singular value $\overbar{\nu}_{1,k}$ of $\Xb_k$ so the optimization problem is not overly restricted by the object space constraints. For further computational efficiency, if the algorithm reaches a point where all the angle-constraint slack variables are zero, it will stop early and add to the current basis a normalized version of the current iteration's intermediate solution. 

Next, if we specify the convex optimization subproblem \eqref{eq:cvx_subprob} for \eqref{DivasOptHeur_equiv_form}, then it becomes
\begin{equation}\label{DivasOptCCSlackLoad}
\begin{aligned} 
& \min_{\vb\astar}& - 2\vb_0^{\top}\left(\sum_{k\in \ib} \check{\Vb}_k \check{\Vb}_k^{\top}\right)\vb\astar + \vb_0^{\top}\left(\sum_{k\in \ib} \check{\Vb}_k \check{\Vb}_k^{\top}\right)\vb_0 & + \tau_t \sum_{k=1}^{2K+2} s_k & & \\
& s.t. & \vb\astart\vb\astar - 2 \frac{\vb_0^{\top}\check{\Vb}_k \check{\Vb}_k^{\top}\vb\astar}{\cos^2\left(\hat{\phi}_k\right)} + \frac{\vb_0^{\top}\check{\Vb}_k \check{\Vb}_k^{\top}\vb_0}{\cos^2\left(\hat{\phi}_k\right)} &\leq s_k & \forall k &\in \ib \\
& & \frac{\vb\astart\check{\Vb}_k \check{\Vb}_k^{\top}\vb\astar}{\cos^2\left(\hat{\phi}_k\right)} - 2\vb_0^{\top}\vb\astar + \vb_0^{\top}\vb_0 &\leq s_k  & \forall k &\in \ib^c \\
& & \vb\astart\Xb_k^{\top}\Xb_k\vb\astar - 2 \frac{ \vb_0^{\top}\Xb_k^{\top}\check{\Ub}_k \check{\Ub}_k^{\top} \Xb_k \vb\astar}{\cos^2\left(\hat{\psi}_k\right)} + \frac{\vb_0^{\top}\Xb^{\top}_k\check{\Ub}_k \check{\Ub}_k^{\top}\Xb_k\vb_0}{\cos^2\left(\hat{\psi}_k\right)} &\leq s_{K+k}/\overbar{\nu}_{1,k} & \forall k &\in \ib \\
& & 1 - 2 \vb_0^{\top} \vb\astar +  \vb_0^{\top}\vb_0 &\leq s_{2K+1} & & \\
& & \vb\astart \vb\astar - 1 &\leq s_{2K+2} & & \\
& & \mathfrak{V}_{\jb}^{\top} \vb\astar &= \0 & \forall \jb &\supseteq \ib.
\end{aligned}
\end{equation}
This problem is in fact a convex optimization problem with linear objective function and convex quadratic and linear constraints, which can be efficiently solved by several convex optimization solvers, including interior-point methods.
In our implementation, we use a MATLAB package associated with a default solver, called CVX from \cite{cvx, gb08} to model the subproblem \eqref{DivasOptCCSlackLoad}.
If we use CVX's default solver, SDPT3, then our algorithm runs in about 30 minutes on the mortality data example from Section~\ref{sec:data} on the authors' laptop with the default solver. 
Larger data sets like the breast cancer genomics example can take considerably longer, but substantial speed-ups are possible with Mosek and other commercial solvers.

To terminate our algorithm, one can look at the objective value of \eqref{DivasOptHeur_equiv_form} to see if it is actually improved through iterations.
If after, e.g., five consecutive iterations, the objective values do not significantly improved, then we can terminate it.
Alternatively, we can also look at the quality of the final solution to see if it is reasonable to terminate the algorithm or not.

\vspace{1ex}
\noindent\textbf{Experiments.}
Figure~\ref{toyexamplestep2} shows the iterative process of the optimization problem for the synthetic data example displayed in Figure~\ref{3blocktruth}. On each panel, the horizontal axis represents the number of iterations and the vertical axis represents the angles in degrees to the panel's respective estimated signal subspace in trait space. Each horizontal green line represents the trait space perturbation angle bound. The blue paths represent the angles to the low rank approximations of trait spaces at each iteration. The red dashed paths represent the angle to the true trait subspaces at each iteration, which are known since this is a synthetic data set. Note that in the right panel in the first row, the initial condition is infeasible for the X3 data block's angle perturbation bound, demonstrating how the algorithm can use the flexibility afforded by the slack variables to explore the space before choosing a final solution. Since the objective function is trying to minimize the total squared cosine of all included blocks, the angle with one data block may increase, as in the top middle panel, if it means the angle with other data blocks decreases.

\begin{figure}[ht!]
\begin{center}
\begin{subfigure}[b]{1\textwidth}
\centering
\includegraphics[scale=0.3]{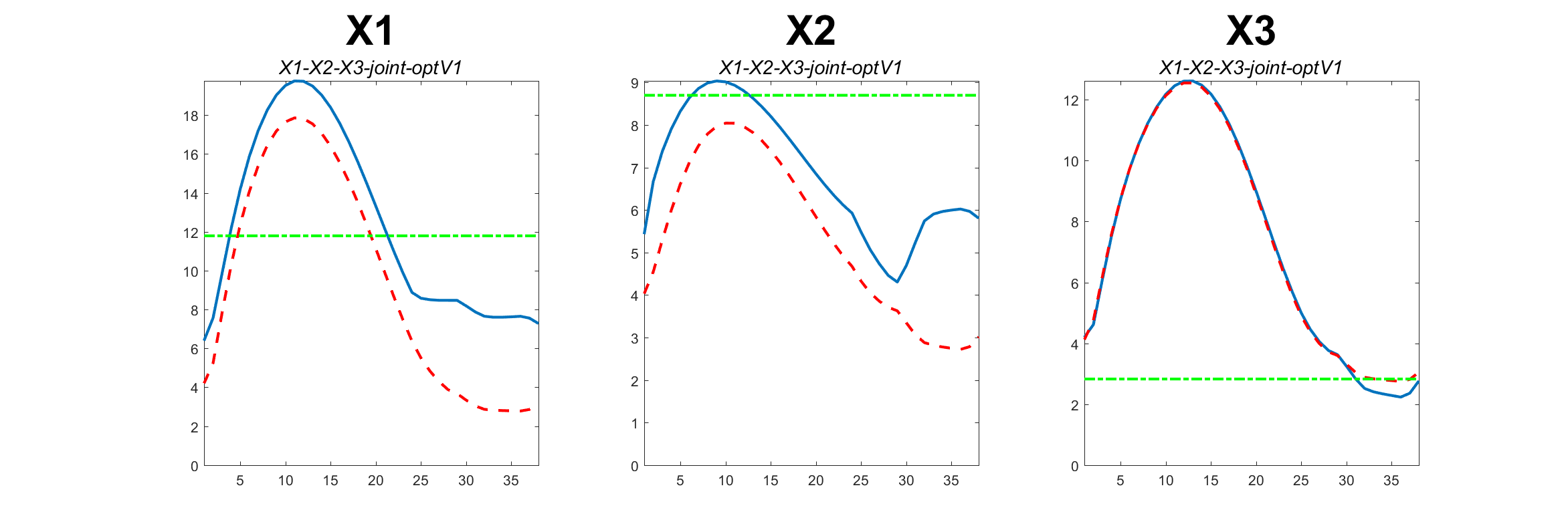}
\end{subfigure}
\begin{subfigure}[b]{1\textwidth}
\centering
	\includegraphics[scale=0.3]{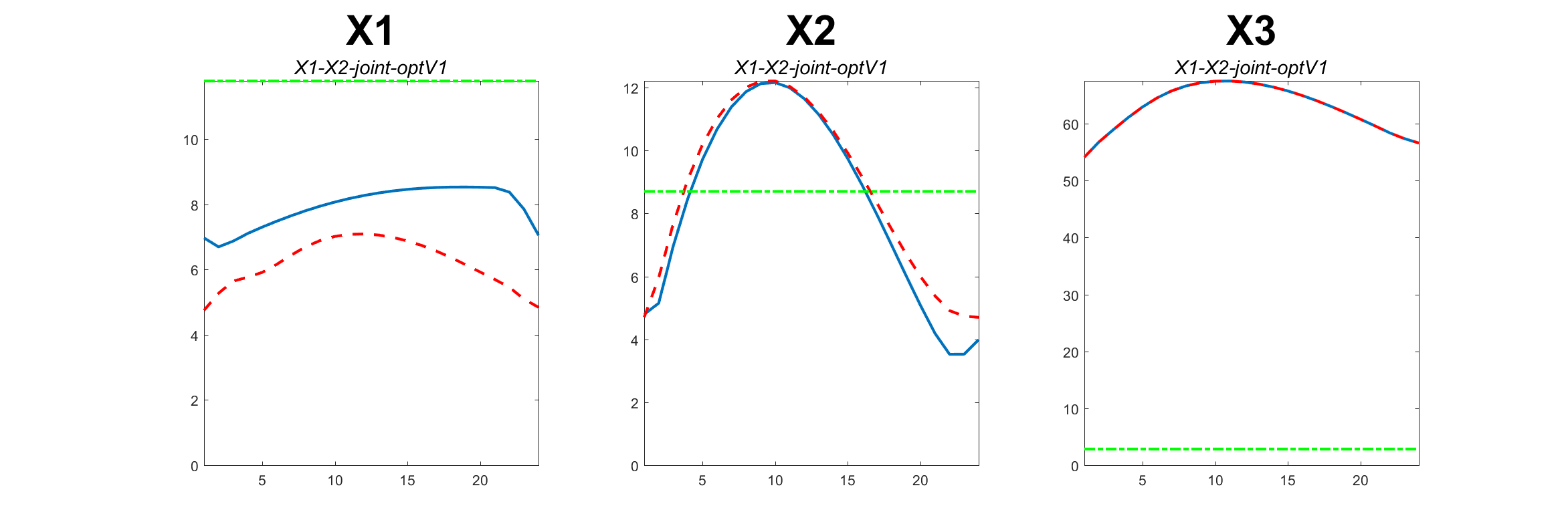}
\end{subfigure}
\begin{subfigure}[b]{1\textwidth}
\centering
	\includegraphics[scale=0.3]{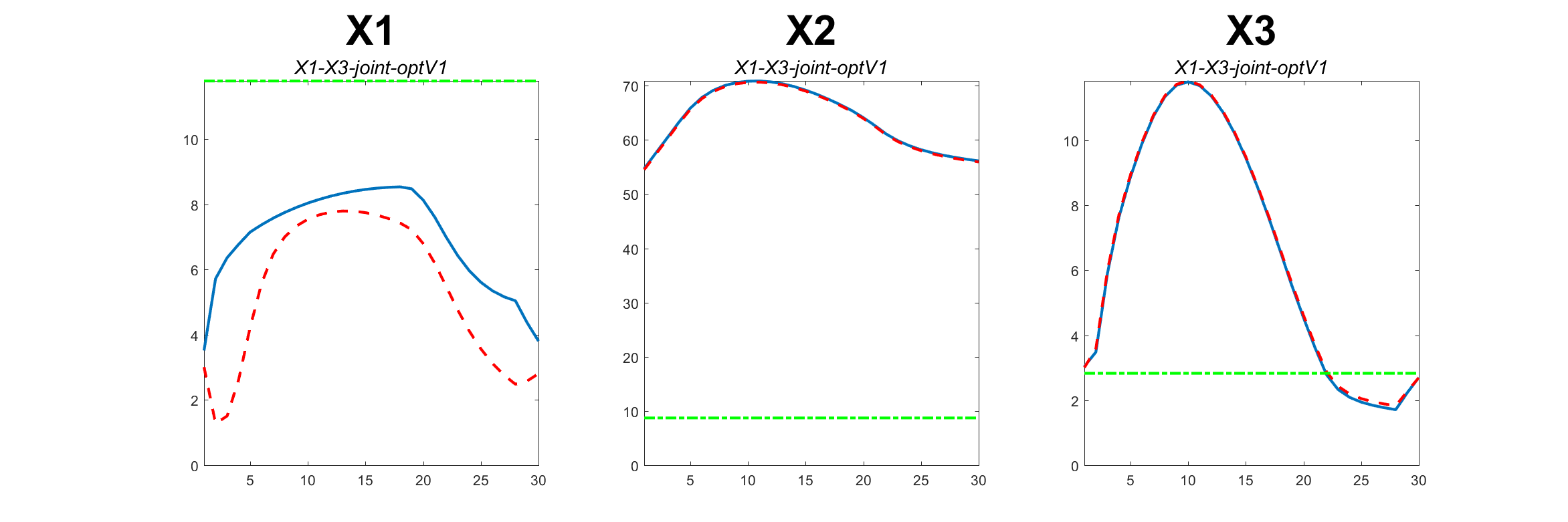}
\end{subfigure}
\begin{subfigure}[b]{1\textwidth}
\centering
	\includegraphics[scale=0.3]{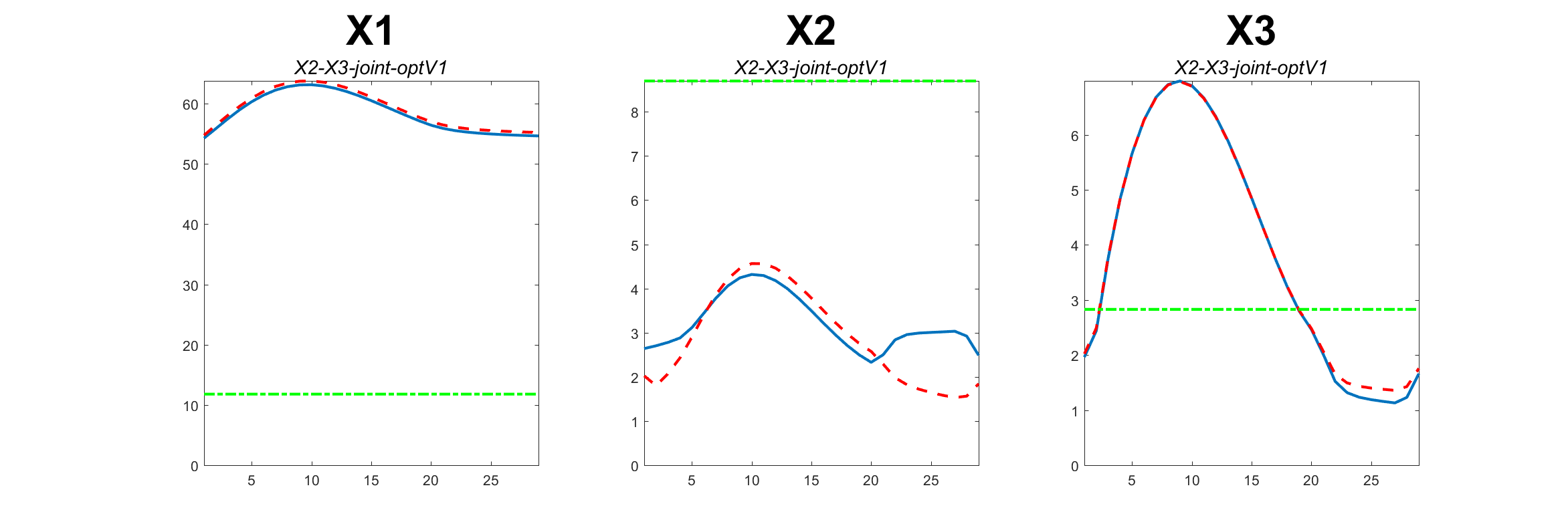}
\end{subfigure}
\caption{Iterative steps of the sequential optimizations locating joint structure between each possible combination of data blocks for the synthetic data from Figure~\ref{3blocktruth}. The horizontal axes represent iterations of the optimization problem and the vertical axes represent angles in degrees. Colored paths show progression of angles between the candidate direction and $\TS(\hat{\Ab}_k)$ (blue) and $\TS(\Ab_k)$ (red). Perturbation angle bounds $\hat{\phi}_k$ shown as green horizontal lines. From top to bottom: three-way joint, joint between blocks 1 and 2, joint between blocks 1 and 3, joint between blocks 2 and 3.}\label{toyexamplestep2}
\end{center}
\end{figure}


\end{document}